\documentclass[12pt,a4paper]{amsproc}

\usepackage[margin=2cm]{geometry}
\usepackage{amsmath,amssymb,amsthm}
\usepackage{mathrsfs}
\usepackage{multirow}
\usepackage{diagbox}
\usepackage{graphicx}
\usepackage{color}

\newtheorem{theorem}{Theorem}
\newtheorem{corollary}{Corollary}
\newtheorem{lemma}{Lemma}
\newtheorem{proposition}{Proposition}
\newtheorem{assumption}{Assumption}
\theoremstyle{remark}
\newtheorem{remark}{Remark}
\newtheorem{example}{Example}

\def\R{\mathbb{R}}
\def\1{\mathbf{1}}
\def\ve{\varepsilon}
\def\hve{\hat{\ve}}

\def\Q{\rho}

\def\htheta{\hat{\theta}}

\def\Femp{\mathbb{F}}
\def\hFemp{\hat{\Femp}}

\def\hR{\hat{R}}

\def\op{o_P(1)}

\def\opn{o_P(n^{-1/2})}

\def\RR{\mathscr{R}}

\def\Lcal{\mathcal{L}}

\def\Zint{\mathbb{Z}}
\def\Nnum{\mathbb{N}}

\def\Dset{\mathbb{D}}

\DeclareMathOperator*{\argmin}{arg\,min}

\numberwithin{equation}{section}

\allowdisplaybreaks

\begin{document}


\title[Regularization parameter selection]{Regularization parameter
  selection in indirect regression by residual based bootstrap}

\begin{abstract}
Residual-based analysis is generally considered a cornerstone of
statistical methodology. For a special case of indirect regression, we
investigate the residual-based empirical distribution function and
provide a uniform expansion of this estimator, which is also shown to
be asymptotically most precise. This investigation naturally leads to
a completely data-driven technique for selecting a regularization
parameter used in our indirect regression function estimator. The
resulting methodology is based on a smooth bootstrap of the model
residuals. A simulation study demonstrates the effectiveness of our
approach.
\end{abstract}

\author[N.\ Bissantz, J.\ Chown and H.\ Dette]{Nicolai Bissantz,
  Justin Chown$^*$ and Holger Dette}

\thanks{$^*$ Correspondences may be addressed to Justin Chown
  (justin.chown@ruhr-uni-bochum.de). \\
  {\em Ruhr-Universit\"at Bochum,
    Fakult\"at f\"ur Mathematik, Lehrstuhl f\"ur Stochastik, 44780
    Bochum, DE}}

\maketitle

\bigskip



\section{Introduction}
\label{intro}

In many experiments one is only able to make indirect observations of
the physical process being observed. Important quantities that are of
interest to the study are not directly available for statistical
inference in such so-called {\sl inverse problems}, but images of
these quantities under some transformation such as a convolution can
be used instead. Here, we consider an inverse regression model, i.e.\
observing a signal of interest from indirect observations
\begin{equation} \label{modeleq}
Y_j = \big[K\theta\big](x_j) + \ve_j,
 \qquad j = -n,\ldots,n,
\end{equation}
where $K$ is an operator specifying convolution of the true underlying
regression $\theta$ with a distortion function $\psi$, i.e.\
\begin{equation*}
\big[K\theta\big](x_j) = \int_{-1/2}^{1/2} \theta(u)\psi(x_j - u)\,du.
\end{equation*}
The resulting function $K\theta$ can be viewed as a distorted
regression function. We assume that $\theta$ is a smooth periodic
function, which is a common assumption taken in many inverse problems
and further discussed below. As in pages 49-51 of Tsybakov (2009),
this means the Fourier coefficients of $\theta$ are assumed to satisfy
a crucial technical summability requirement (see Section
\ref{mainresults} for further details).

We will assume that $\psi$ is known and behaves like a probability
density function on the interval $[-1/2,\,1/2]$, i.e.\ $\psi$ is
positive--valued on the interval $[-1/2,\,1/2]$ and integrates to one.
Later, we will specify further technical requirements regarding
$\psi$. However, to ensure the convolution operation is well--defined,
it is clear that $\psi$ must be periodically extended to the intervals
$[x - 1/2,\,x + 1/2]$ for each $x \in [-1/2,\,1/2]$. The covariates
$x_j$ in model \eqref{modeleq} are uniformly distributed design points
in the interval $[-1/2,\,1/2]$, i.e.\ $x_j = j/2n$, $j = -n,\ldots,n$,
and the errors $\ve_j$ are assumed to be independent, have mean equal
to zero and have the common distribution function $F$. Note, the
assumptions given above only guarantee that model \eqref{modeleq} is a
well--defined indirect regression model, where $\theta$ is
identifiable (see, for example, Cavalier and Golubev, 2006, or Mair
and Ruymgaart, 1996).

Statistical inverse problems have received a great amount of attention
for the construction of estimators for various densities and indirect
regression models. In particular, early work considers properties of
estimators in a range of important statistical inverse
problems. Important examples are Masry (1991), who investigates
estimators of a multivariate density function in an
errors-in-variables model using a deconvolution technique; Fan (1991),
who derives optimal rates of convergence for density estimators in
these models, and Masry (1993), who investigates estimators of a
smooth multivariate regression function using deconvolution techniques
for the case of regression function estimation with contaminated
covariates.

A second sequence of publications about statistical inverse problems
such as those considered in \eqref{modeleq} yields a better
understanding of the asymptotic properties of the estimators from a
theoretical perspective. Important results are due to Mair and
Ruymgaart (1996), who consider estimators of the indirect regression
function in a flexible model based on Hilbert scales. This includes
the popular case of Sobolev classes, where these authors describe
general regularization approaches for operator inversion and show that
the considered estimators are in fact minimax optimal. Cavalier and
Tsybakov (2002) investigate an indirect heteroscedastic regression
model and prove minimax optimality of their estimators. Moreover,
Cavalier (2008) surveys the available literature for deconvolution
estimators and provides minimax rates of reconstructions in several
models including \eqref{modeleq}.

More recently, statistical testing and model selection properties have
been considered in statistical inverse problems of the type
\eqref{modeleq}. Bissantz and Holzmann (2008) provide an overview for
constructing confidence intervals and confidence bands in univariate
statistical inverse problems. Later, Proksch, Bissantz and Dette
(2015) generalize the univariate case studied by the previous authors
and construct confidence bands for an indirect regression function of
multiple covariates. Marteau and Math\'e (2014) consider the problem
of testing for distorted signals using general regularization schemes.

All of the deconvolution estimators investigated in the previously
mentioned articles are based on projections of the data and result in
kernel-type estimators that depend on some kind of regularization
parameter. This quantity is analogous to the bandwidth found in the
usual nonparametric function estimators. Data--driven selection of
this parameter is an important problem that we want to more closely
examine in this article. Techniques for choosing the sequence of
regularization parameters generally focus on choosing a suitable
estimator of the integrated mean squared error of an indirect
regression estimator or choosing some other related quantity (see
Section \ref{apps}). Cavalier and Golubev (2006) make a particularly
important contribution to this problem, where these authors
investigate the integrated mean squared error of indirect regression
estimators and propose a suitably penalized quantity based on a
threshold of this important estimation performance metric. The authors
call this a {\em risk hull} approach because of the resulting
bowl-shaped objective function used for choosing their parameter
sequence. From another perspective, we can consider potential
bootstrap approaches to this problem, where we instead calculate the
integrated mean squared error of a bootstrap version of the indirect
regression estimator. The bootstrap method for choosing the
regularization parameter sequence investigated here appears to
be particularly promising when compared with the risk hull approach
previously mentioned.

In this article, we provide a statistical methodology for selecting a
best fitting (most feasible) regression estimator from a sequence of
function estimators based on observations from model \eqref{modeleq}
using the resulting model residuals constructed from the estimator
$\htheta$, see \eqref{thetahat}:
\begin{equation*}
\hve_j = Y_j - \big[K\htheta\big](x_j),
 \qquad j=-n,\ldots,n.
\end{equation*}
Many statistical procedures are residual--based, including the
bootstrap methodology for selecting the regularization that we
investigate. This requires that we first study the distribution
function $F$ of the model errors, which is generally unknown and must
be estimated.

To the best of our knowledge, this object has not been
studied before with respect to statistical deconvolution in a
completely nonparametric setting. We form an estimator of $F$ using
the empirical distribution function of the model residuals:
\begin{equation*}
\hFemp(t) = \frac{1}{2n + 1} \sum_{j = -n}^{n} \1\big[\hve_j \leq t\big]
 = \frac{1}{2n + 1} \sum_{j = -n}^{n}
 \1\big[Y_j - \big[K\htheta\big](x_j) \leq t\big],
 \qquad t \in \R.
\end{equation*}
The estimator $\htheta$ is shown here to be a suitable estimator of
$\theta$ such that we can study the limiting behavior of
$\hFemp$, which is new. In addition, this work reveals stronger
conditions regarding the smoothness of $\theta$ are required for
$\hFemp$ to be a consistent estimator of $F$ that have not appeared
before in the literature. Hence, any residual--based inference
procedure relying on $\hFemp$ would also require this stronger
smoothness condition, e.g.\ Kolmogorov-Smirnov-type and
Cram\'er-von-Mises-type statistics.

Studying these problems requires new results concerning the estimator
$\htheta$ and its bootstrap analog. The literature on statistical
deconvolution problems is vast, and, hence, some results will be
familiar. In particular, we show the estimator $\htheta$ has a strong
uniform rate of consistency for the function $\theta$ that is
analogous to the already known minimax optimal rate of convergence
(see Theorem \ref{thmthetahatuniformrate} in Section
\ref{mainresults1} and Remark \ref{remhnoptimal} in Section
\ref{apps}). There are also many results in the literature on
residual--based empirical distribution functions for direct regression
models; for example, uniform consistency and asymptotic
optimality. We show the estimate $\hFemp$ satisfies both of these
properties (see Theorem \ref{thmFhatExpan} and Remark
\ref{remFhatEfficient} in Section \ref{mainresults1}). The
residual-based empirical distribution function for a wide class of
semiparametric direct regression models are studied by M\"uller,
Schick and Wefelmeyer (2007), and we derive comparable results for the
indirect regression model \eqref{modeleq}.

We have organized the remainder of this article as follows. Further
notation and the estimation method are introduced in Section
\ref{mainresults}, and the asymptotic results for the estimators
$\htheta$ and $\hFemp$ may be found in Section \ref{mainresults1}. In
Section \ref{apps}, we consider the problem of finding an optimal
regularization parameter for the estimator $\htheta$. Here we provide
a rule-of-thumb approach that is in the spirit of Silverman (1986),
and, in Section \ref{boot}, we develop a data--driven approach for
selecting this parameter using a smooth bootstrap of the model
residuals that is in the spirit of Neumeyer (2009). We conclude the
article with a numerical study in Section \ref{sims}, which indicates
good finite sample performance of the proposed data-driven
regularization against the theoretically optimal regularization, and
we consider a comparative technique for choosing the regularization
for spectral cut-off estimates (a special case of our approach)
proposed by Cavalier and Golubev (2006) as an example in Section
\ref{comparemethods}. All of the proofs of our results may be found in
Section \ref{details}.


\section{Estimation in the indirect regression model}
\label{mainresults}

Let us begin with the space of square integrable functions
$\Lcal_2([-1/2,\,1/2])$ with domain $[-1/2,\,1/2]$. This function
space has the well known and countable orthonormal basis
\begin{equation*}
\Big\{e^{i2\pi k x}\,:\,x\in[-1/2,\,1/2]\Big\}_{k \in \Zint}.
\end{equation*}
In order to construct an  estimator for the function $\theta$ we will
need to restrict $\theta$ to a smooth class of functions from
$\Lcal_2([-1/2,\,1/2])$. This means we only consider functions
$\theta$ that are weakly differentiable in $\Lcal_2([-1/2,\,1/2])$.

For clarity, we will now introduce some notation. Let $d \in
\Nnum$. We will call $q^{(i)}$, $1 \leq i \leq d$, a weak derivative
of $q$ in $\Lcal_2([-1/2,\,1/2])$ of order $i$, if $q^{(i)} \in
\Lcal_2([-1/2,\,1/2])$ and $q^{(i)}$ satisfies
\begin{equation*}
\int_{-1/2}^{1/2} q(x)\frac{d^i}{dx^i}\phi(x)\,dx
 = (-1)^i\int_{-1/2}^{1/2} q^{(i)}(x)\phi(x)\,dx,
\end{equation*}
for every infinitely differentiable function $\phi$ with support
$[-1/2,\,1/2]$ that have evaluations of $\phi$,
$(d^{i'}\phi)/(dx^{i'})$, $i' = 1,\ldots,i$, at $1/2$ and $-1/2$ equal
to zero. The corresponding space of functions is the Sobolev space
$\mathcal{W}^{2,d}([-1/2,\,1/2])$, where
\begin{align*}
\mathcal{W}^{2,d}\big([-1/2,\,1/2]\big)
 &= \bigg\{ q \in \Lcal_2([-1/2,\,1/2]) \,:\,
 q^{(1)},\ldots,q^{(d)} \in \Lcal_2([-1/2,\,1/2])
 \bigg\} \\
&= \bigg\{ q \in \Lcal_2([-1/2,\,1/2]) \,:\,
 \sum_{k = -\infty}^{\infty} \big(1 + k^{2}\big)^{d} |\Q(k)|^2 < \infty
 \bigg\}.
\end{align*}
Here $\{\Q(k)\}_{k \in \Zint}$ are the Fourier coefficients of $q$:
\begin{equation*}
\Q(k) = \int_{-1/2}^{1/2} q(u) e^{-i2\pi k u}du,
\qquad k \in \Zint.
\end{equation*}
Replacing $d$ with a positive real number motivates considering
smoothness orders $s > 0$, i.e.\ $\mathcal{W}^{2,s}([-1/2,\,1/2])$ is
defined exactly as $\mathcal{W}^{2,d}([-1/2,\,1/2])$ is defined above
but now with $s$ in place of $d$. We will require that $\theta$
satisfies a stronger series condition than that stated for
$\mathcal{W}^{2,d}([-1/2,\,1/2])$ above. Following Condition $C_1$ in
Politis and Romano (1999), which is similar to a condition imposed in
Watson and Leadbetter (1963), we will restrict $\RR_{s}$ to a subspace
of $\mathcal{W}^{2,s}([-1/2,\,1/2])$:
\begin{equation*}
\RR_{s} = \bigg\{ q \in \mathcal{W}^{2,s}\big([-1/2,\,1/2]\big) \,:\,
 \sum_{k = -\infty}^{\infty} |k|^{s}|\Q(k)| < \infty \bigg\}.
\end{equation*}
Note, $\theta \in \RR_{s}$ implies a restriction on the Fourier
coefficients $\{\Theta(k)\}_{k \in \Zint}$ of $\theta$, which are
defined similarly to the Fourier coefficients $\{\Q(k)\}_{k \in
  \Zint}$ above.

Another important note regarding the Fourier basis $\{\exp(i2\pi
kx):x \in [-1/2,\,1/2]\}_{k \in \Zint}$ is that it also decomposes
the operator $K$ into a singular value decomposition along each of the
orthonormal basis functions, where now we only need to consider the
Fourier coefficients $\{\Psi(k)\}_{k \in \Zint}$ of the distortion
function $\psi$, which are defined similarly to the Fourier
coefficients $\{\Q(k)\}_{k \in \Zint}$ above. Much of the research in
the area of deconvolution problems has focused on two important
cases. The first case is that of the so--called ordinarily smooth
distortion functions, which means assuming the Fourier coefficients
$\{\Psi(k)\}_{k \in \Zint}$ decay at a polynomic rate: there is some
$b > 0$ such that $|\Psi(k)| \sim |k|^{-b}$. Here ``$\sim$'' denotes
asymptotic similarity. Under this assumption, we
can construct an estimator $\htheta$ for $\theta$ whose strong uniform
consistency rate is comparable, albeit worse, to the rates expected in
the usual nonparametric regression case, and we can show the estimator
$\hFemp$ is both root--$n$ consistent for $F$, uniformly in $t \in
\R$, and $\hFemp$ is asymptotically most precise. The second case is
that of the so--called super smooth distortion functions, which means
assuming the Fourier coefficients $\{\Psi(k)\}_{k \in \Zint}$ decay at
an exponential rate, e.g.\ $|\Psi(k)| \sim \exp(-|k|^{b})$. Under this
assumption, the resulting indirect regression estimator has only a
strong uniform consistency rate that is polynomic in the logarithm of
$n$, which we expect is too slow for us to maintain the root--$n$
consistency of $\hFemp$. Throughout this article, we will therefore
focus on the first case of ordinarily smooth distortion functions
$\psi$, and we work with a similar assumption as (1.4) of Fan (1991):
\begin{assumption} \label{assumpPsi}
There are finite constants $b > 0$, $\Gamma > 0$, and $0 < C_{\Psi} <
C_{\Psi}^*$ such that for every $|k| > \Gamma$ the Fourier
coefficients $\{\Psi(k)\}_{k \in \Zint}$ of $\psi$ satisfy $C_{\Psi} <
|k|^b|\Psi(k)| < C_{\Psi}^*$.
\end{assumption}

\begin{example} \label{expsi}
Suppose $\psi$ is known to be the standard Laplace density function
restricted to the interval $[-1/2,\,1/2]$, i.e.\
\begin{equation*}
\psi(x) = \frac12 e^{-|x|} \bigg\slash
 \bigg[ \int_{-1/2}^{1/2} \frac12 e^{-|x|}\,dx \bigg]
 = \frac12 e^{-|x|} \Big\slash 
 \Big[ 1 - e^{-1/2} \Big],
 \qquad x \in [-1/2,\,1/2].
\end{equation*}
The Fourier coefficients $\{\Psi(k)\}_{k \in \Zint}$ are then given by
\begin{equation*}
\Psi(k) = \frac{1}{1 + 4\pi^2k^2},
 \qquad k \in \Zint.
\end{equation*}
Then Assumption \ref{assumpPsi} is satisfied for the choices $b = 2$,
$\Gamma = 1$, $C_{\Psi} = (1 + 4\pi^2)^{-1}$ and $C_{\Psi}^* =
(4\pi^2)^{-1}$.
\end{example}

Recall that we use a uniform fixed design on the interval
$[-1/2,\,1/2]$. Writing $Q$ for the conditional distribution of a
response $Y$ given a fixed design point $x$ results in the equivalence
$Q(y\,|\,x) = P_x(Y \leq y)$, where $P_x$ denotes the distribution of
$Y$ depending on $x$, which is not random. It follows that we can
write the Fourier coefficients $\{R(k)\}_{k \in \Zint}$ of $K\theta$ as
\begin{equation} \label{Rk}
R(k) = \int_{-1/2}^{1/2}\int_{-\infty}^{\infty}
 ye^{-i2\pi k x}\,Q(dy\,|\,x)\,dx,
 \qquad k \in \Zint.
\end{equation}
The double integral in the right-hand side of \eqref{Rk} is an average,
and we can form an estimator for it from an empirical average using data
$(x_j,Y_j)$, $j = -n,\ldots,n$, and obtain
\begin{equation*}
\hR(k) = \frac{1}{2n + 1} \sum_{j = -n}^{n} Y_j e^{-i2\pi k x_j},
\qquad k \in \Zint.
\end{equation*}

To recover $\theta$ from the convolution $K\theta$, we will make use
of the convolution theorem for Fourier transformation: $R(k) =
\Theta(k)\Psi(k)$, $k \in \Zint$. Since the periodic extension of
$\psi$ is positive--valued on any interval $[x - 1/2,\,x + 1/2]$, $x
\in [-1/2,\,1/2]$, it follows that $\{\Psi(k)\}_{k \in \Zint}$ is
bounded away from zero in absolute value on any bounded region
$\mathcal{Z} \subset \Zint$, and, hence, $\Psi^{-1}$ is well--defined
(see, for example, the discussion on preconditioning on page 1425 of
Mair and Ruymgaart, 1996). Observing that Fourier transformation
reduces convolution to multiplication, we can exploit the Fourier
inversion formula by writing
\begin{equation*}
\theta(x) = \sum_{k = -\infty}^{\infty} \frac{R(k)}{\Psi(k)} e^{i2\pi k x},
 \qquad x \in [-1/2,\,1/2].
\end{equation*}
To plug--in our estimated Fourier coefficients $\{\hR(k)\}_{k \in \Zint}$
for the Fourier coefficients $\{R(k)\}_{k \in \Zint}$, we need to
control the random fluctuations that occur at high frequency spectra,
i.e.\ the inversion of the operator $K$ in \eqref{modeleq} requires
regularization; see Cavalier and Golubev (2006) for a very clear
discussion on regularization and ill-posedness.

Politis and Romano (1999) introduce spectral smoothing to control
these fluctuations in higher frequencies, which amounts to
regularizing the inversion operator in same spirit as Mair and
Ruymgaart (1996). To explain the idea consider the ratio
$|\hR(k)|/|\Psi(k)|$, which becomes large as $k$ increases. The idea
is to utilize lower frequencies and dampen the contributions from
higher frequencies by introducing a sequence of weights. The most
striking difference between the approaches taken by Politis and Romano
(1999) and Mair and Ruymgaart (1996) is the previous authors require
the regularization to preserve the fundamental Fourier frequency,
i.e.\ the regularization must be equal to one around some neighborhood
of the zeroth Fourier frequency. This approach to regularization is
simple to specify for applications and leads to desired optimality
properties.

Let us now introduce some notation. Write $\{h_n\}_{n \geq 1}$ for a
regularizing sequence that satisfies $h_n \to 0$, as $n \to
\infty$. Consider smoothing kernel functions similar to those used in
typical nonparametric function estimators, i.e.\ maps $x \mapsto
h_n^{-1}\delta(x/h_n)$, where $\delta$ is a suitably constrained
probability density function. Politis and Romano (1999) observe the
Fourier transform of a smoothing kernel $h_n^{-1}\delta(x/h_n)$ takes
the form $\Lambda(h_nk)$, where $\Lambda$ is the Fourier transform of the
desired kernel function $K_{\Lambda}$. This means the Fourier
transform of a smoothing kernel depends on $n$ only through the
regularizing sequence $\{h_n\}_{n \geq 1}$ by shrinking the Fourier
frequency from $k$ to $h_nk$. We will require our smoothing kernel to
have a Fourier transform $\Lambda$ that satisfies the following
general assumption:
\begin{assumption} \label{assumplambda}
The region $I = \{k \in \Zint\,:\,|k| \leq M\}$ exists for some
integer $M \geq 1$ such that $\Lambda(k) = 1$, when $k \in I$, and
$|\Lambda(k)| \leq 1$, otherwise. Let $\Lambda$ satisfy
$\int_{-\infty}^{\infty} |u|^{b}|\Lambda(u)|\,du < \infty$, where $b >
0$ is the degree of ill-posedness introduced in Assumption
\ref{assumpPsi}.
\end{assumption}

We will use the order notation $O(a_n)$ for a sequence $\{b_n\}_{n
\geq 1}$ satisfying $a_n^{-1}b_n \to L$ for some finite constant $L$,
and we write $o(a_n)$ when $L = 0$. Analogously, we will write
$O_P(a_n)$ and $o_P(a_n)$ when these statements hold with probability
tending to one as the sample size $n$ increases. Assumption
\ref{assumplambda} ensures only the estimation bias has a desirable
rate of convergence: the order $O(h_n^{s})$, when $\theta \in
\RR_{s}$. This is comparable to the direct estimation setting by
sufficiently high-order kernels or the so--called ``superkernels''
(see, for example, the discussion on page 3 of Politis and Romano,
1999). The idea of restricting the choice of the smoothing kernel
function based on obtaining a suitable rate of convergence in the
estimation bias dates all the way back to Parzen (1962).

An estimator of $\theta$ is given by a kernel smoother:
\begin{equation} \label{thetahat}
\htheta(x) = \sum_{k = -\infty}^{\infty} \Lambda(h_nk)
 \frac{\hR(k)}{\Psi(k)} e^{i2\pi k x}
 = \frac{1}{2n + 1} \sum_{j = -n}^{n} Y_j
 W_{h_n}\big(x - x_j\big),
 \qquad x \in [-1/2,\,1/2],
\end{equation}
where the smoothing kernel $W_{h_n}$ is given by
\begin{equation*}
W_{h_n}\big(x - x_j\big) = \sum_{k = -\infty}^{\infty}
 \frac{\Lambda(h_nk)}{\Psi(k)}
 \exp\Big(i2\pi k \big(x - x_j\big)\Big).
\end{equation*}
The smoothing kernel $W_{h_n}$ is sometimes called a deconvolution
kernel (see, for example, Birke, Bissantz and Holzmann, 2010).


\subsection{Asymptotic results for the deconvolution estimator and the
empirical distribution function of the residuals}
\label{mainresults1}

Our first result specifies the asymptotic order of the bias of
$\htheta$.

\begin{lemma} \label{lemhthetabias}
Let $\theta \in \RR_{s}$, with $s \geq 1$, and let Assumptions
\ref{assumpPsi} and \ref{assumplambda} hold. Then, for any
regularizing sequence $\{h_n\}_{n \geq 1}$ satisfying $h_n \to 0$ and
$nh_n^{b + 1} \to \infty$, as $n \to \infty$, we have
\begin{equation*}
\sup_{x \in [-1/2,\,1/2]} \Big|
 E\big[\htheta(x)\big] - \theta(x) \Big|
 = O\big(h_n^{s} + (nh_n^{b + 1})^{-1}\big).
\end{equation*}
\end{lemma}

The asymptotic order of the bias of $\htheta$ is impacted by the
degree of ill-posedness of the inverse problem, which we expect can be
made negligible by choice of regularization parameters $\{h_n\}_{n
  \geq 1}$. In the following result, we observe this detrimental
effect in the asymptotic order of consistency as well.
\begin{lemma} \label{lemthetahatconsistency}
Let $\theta \in \RR_{s}$, with $s \geq 1$, and let Assumptions
\ref{assumpPsi} and \ref{assumplambda} hold. Assume that $\Lambda$
additionally satisfies $\int_{-\infty}^{\infty}\,|u|^{b +
  1}|\Lambda(u)|\,du < \infty$ and the random variables
$Y_{-n},\ldots,Y_{n}$ have a finite absolute moment of order $\kappa >
2 + 1/b$. Finally, let the regularizing sequence $\{h_n\}_{n \geq 1}$
satisfy $h_n \to 0$ such that $(nh_n^{2b + 1})^{-1}\log(n) \to
0$, as $n \to \infty$. Then
\begin{equation*}
\sup_{x \in [-1/2,\,1/2]} \Big| \htheta(x) - E\big[\htheta(x)\big] \Big|
 = O\Big(\big(nh_n^{2b + 1}\big)^{-1/2}\log^{1/2}(n)\Big),
 \qquad\text{a.s.}
\end{equation*}
\end{lemma}

The two lemmas above imply that we can obtain a strong uniform
rate of convergence of the estimator $\htheta$ for $\theta$ by
choosing a regularizing sequence $\{h_n\}_{n \geq 1}$ that balances
the asymptotic orders of both the bias and consistency, i.e.\
\begin{equation} \label{bw} 
h_n = O\big( n^{-1/(2s + 2b + 1)}\log^{1/(2s + 2b + 1)}(n) \big).
\end{equation}
For this choice of regularizing parameters, we have $(nh_n^{b +
  1})^{-1} = o(h_n^{s})$, which implies the bias of $\htheta$ has the
order $O(h_n^{s})$. Note, Lemma \ref{lemthetahatconsistency} requires
the responses to have a finite moment of order larger than $2 + 1/b$,
which is only a sufficient condition. One can easily show that $\kappa
> 2 + 1/(s + b)$ is necessary when $\{h_n\}_{n \geq 1}$ satisfies
\eqref{bw}, which is more reasonable for situations when $b \to 0$. We
now state the uniform rate of convergence of $\htheta$ for $\theta$ when
the parameter sequence $\{h_n\}_{n \geq 1}$ satisfies \eqref{bw}, and
two additional properties of the estimator $\htheta$.

\begin{theorem} \label{thmthetahatuniformrate}
Let the assumptions of Lemma \ref{lemthetahatconsistency} hold, but
now only requiring $\kappa > 2 + 1/(s + b)$. Choose the regularizing
sequence $\{h_n\}_{n \geq 1}$ to satisfy \eqref{bw}. Then
\begin{equation*}
\sup_{x \in [-1/2,\,1/2]} \Big|\htheta(x) - \theta(x)\Big|
 = O\big(n^{-s/(2s + 2b + 1)}\log^{s/(2s + 2b + 1)}(n)\big),
 \qquad\text{a.s.}
\end{equation*}
If, additionally, $s > (2b + 1)/(2\gamma)$, for some $0 < \gamma \leq
1$, then
\begin{equation*}
\bigg[\sup_{x \in [-1/2,\,1/2]} \Big|\htheta(x) - \theta(x)\Big|
 \bigg]^{1 + \gamma} = o(n^{-1/2}),
 \qquad\text{a.s.}
\end{equation*}
If $\Lambda$ satisfies $\int_{-\infty}^{\infty}\,|u|^{s +
  b - 1/2}|\Lambda(u)|\,du < \infty$, then, for large enough $n$,
\begin{equation*}
\htheta - \theta \in \RR_{s - 1/2,1},
 \qquad\text{a.s.},
\end{equation*}
where $\RR_{s - 1/2,1} = \{q \in \RR_{s - 1/2}\,:\,\|q\|_{\infty} \leq
1\}$ is the unit ball of the metric space
$(\RR_{s - 1/2},\,\|\cdot\|_{\infty})$.
\end{theorem}

\begin{remark} \label{remhthetasmoothnessrate}
The second statement of Theorem \ref{thmthetahatuniformrate} requires
the smoothness index $s$ of the function space $\RR_{s}$ to be larger
than the degree of ill--posedness $b$ of the inverse problem, which is
a stronger requirement than what has appeared in the literature
before. The additional smoothness is simply explained by the
entanglement of the smoothness index $s$ and the degree of
ill--posedness $b$ in the strong uniform consistency rate given in the
first statement of Theorem \ref{thmthetahatuniformrate}: $O(n^{-s/(2s
  + 2b + 1)}\log^{s/(2s + 2b + 1)}(n))$. This entanglement also occurs
for indirect regression estimators satisfying minimax optimality,
where now the integrated mean squared error has the order
$O(n^{-(2s)/(2s + 2b + 1)})$.
\end{remark}

We are now ready to state our main results concerning the estimator
$\hFemp$.

\begin{theorem} \label{thmFhatExpan}
Assume the distribution function $F$ admits a bounded Lebesgue density
function $f$ that is H\"older continuous with exponent $0 < \gamma
\leq 1$, and let $\ve_{-n},\ldots,\ve_{n}$ have a finite absolute
moment of order $\kappa > 2 + 1/(s + b)$. Let the assumptions of
Theorem \ref{thmthetahatuniformrate} be satisfied, with $s > \max\{(2b
+ 1)/(2\gamma),\,3/2\}$. Finally, let the regularizing sequence
$\{h_n\}_{n \geq 1}$ satisfy \eqref{bw}. Then
\begin{equation*}
\sup_{t \in \R} \bigg| \frac{1}{2n + 1} \sum_{j = -n}^{n}
 \Big\{ \1\big[\hve_j \leq t\big] - \1\big[\ve_j \leq t\big]
 - \ve_jf(t) \Big\} \bigg| = \opn.
\end{equation*}
\end{theorem}

\begin{corollary} \label{corFhatGaussian}
Under the conditions of Theorem \ref{thmFhatExpan}, the process
\begin{equation*}
(2n + 1)^{1/2}\{\hFemp(t) - F(t)\}
 = (2n + 1)^{-1/2} \sum_{j = -n}^{n}
 \Big\{\1\big[\ve_j \leq t\big] - F(t) + \ve_jf(t)\Big\} + \op,
\end{equation*}
$t \in \R$, weakly converges to a mean zero Gaussian process
$\{Z(t)\,:\,t \in \R\}$, with covariance function, for $u,v \in \R$,
\begin{equation*}
\Sigma(u,v) = F\big(\min\{u,\,v\}\big) - F(u)F(v)
 + f(u)E\big[\ve\1[\ve \leq v]\big] + f(v)E\big[\ve\1[\ve \leq u]\big]
 + \sigma^2f(u)f(v),
\end{equation*}
writing $\sigma^2 = E[\ve^2]$ and $\ve$ for a generic random variable
with distribution function $F$.
\end{corollary}

\begin{remark} \label{remFhatEfficient}
Model \eqref{modeleq} is a nonparametric regression. The estimator
$\hFemp$ has influence function $\1[\ve \leq t] - F(t) + \ve f(t)$,
where $\ve$ is a generic random variable with distribution function
$F$. If we additionally assume that $F$ has finite Fisher information
for location, it follows that $\hFemp$ is efficient for estimating
$F$, in the sense of H\'ajek and Le Cam, from the results of M\"uller,
Schick and Wefelmeyer (2004).
\end{remark}


\section{Regularization parameter selection and the smooth bootstrap
of residuals}
\label{apps}

We now consider the problem of choosing an appropriate sequence of
regularization parameters $\{h_n\}_{n \geq 1}$ required by the
estimator $\htheta$. Popular approaches in the literature suggest a
practical choice of regularization would be a scheme that minimizes
the integrated mean squared error (IMSE) of $\htheta$. However,
selection of such a parameter can also be viewed as a model selection
problem, where we select the \emph{most feasible} regression model
from a sequence of regression function estimators generated from a
sequence of regularization parameters. In the case of iterative
estimation procedures, a suitable stopping iteration is
sought. Multiscale and related methods based on partial sums of
normalized residuals have been thoroughly investigated in the
literature (see, for example, Gonz\'alez-Manteiga, Martinez-Miranda
and P\'erez-Gonz\'alez, 2004; Bissantz, Mair and Munk, 2006; Bissantz,
Mair and Munk, 2008; Davies and Meise, 2008, and Hotz et al.,
2012). Lepski methodology among others has recently become a popular
approach in this context, where the IMSE of the indirect regression
estimator is replaced by a suitable non-random objective function using
oracle inequalities (see, for example, Goldenshluger, 1999; Cavalier
and Tsybakov, 2002; Math\'e and Pereverzev, 2006; Blanchard and
Math\'e, 2012, and Blanchard, Hoffmann and Rei\ss, 2016). An important
approach for spectral cut-off estimators based on assessing the risk
hull of these estimates is investigated by Cavalier and Golubev
(2006), which was already discussed in the introduction. In contrast
to previous works, we will propose a methodology based on a smooth
bootstrap of the model residuals to form a consistent estimator of the
IMSE of $\htheta$, and, using the perspective of conducting model
selection, we propose choosing the regularization parameter sequence
that minimizes this quantity.

In the following result, we give the asymptotic order of the integrated
variance and the integrated squared bias of the estimator $\htheta$
that will lead to a rule-of-thumb approach for selecting regularization
parameters that approximately minimize the IMSE of $\htheta$.
\begin{proposition} \label{prophthetaVarBiassq}
Let $\theta \in \RR_{s}$, with $s \geq 1$, and let Assumptions
\ref{assumpPsi} and \ref{assumplambda} hold. Assume that
$\ve_{-n},\ldots,\ve_{n}$ have finite variance $\sigma^2$. Then, for any
regularizing sequence $\{h_n\}_{n \geq 1}$ satisfying $h_n \to 0$ such
that both $nh_{n}^{2b + 1} \to \infty$, as $n \to \infty$, and
$(nh_n^{b + 1})^{-1} = o(h_n^{s})$ hold, there are constants
$C_{\Lambda} > 0$ and $C_{R} > 0$ such that
\begin{equation*}
\int_{-1/2}^{1/2} E\Big[\big\{\htheta(x)
 - E\big[\htheta(x)\big]\big\}^2\Big]\,dx
 = C_{\Lambda}\sigma^2(nh_n^{2b + 1})^{-1} + o\big((nh_n^{2b + 1})^{-1}\big)
\end{equation*}
and
\begin{equation*}
\int_{-1/2}^{1/2} \Big\{E\big[\htheta(x)\big] - \theta(x)\Big\}^2\,dx
 = C_{R}h_n^{2s} + o\big(h_n^{2s}\big).
\end{equation*}
\end{proposition}

\begin{remark} \label{remhnoptimal}
From the results of Proposition \ref{prophthetaVarBiassq}, we can
obtain an approximately optimal regularizing sequence, in the sense of
minimizing the IMSE of $\htheta$:
\begin{equation*}
h_{n,opt} \approx \bigg(\frac{2b + 1}{2s}\frac{C_{\Lambda}}{C_{R}}
 \sigma^2 \bigg)^{1/(2s + 2b + 1)} n^{-1/(2s + 2b + 1)}.
\end{equation*}
Consequently, the integrated mean squared error of $\htheta$ is of
the order $O(n^{-(2s)/(2s + 2b + 1)})$. Setting $\epsilon = \epsilon_n =
O(n^{-1/2})$ in Table 1 on page 9 of Cavalier (2008) yields that
$\htheta$ is indeed minimax optimal for estimating $\theta$.
\end{remark}
\smallskip

The conclusion that $\htheta$, formed from a regularizing sequence of
order $O(n^{-1/(2s + 2b + 1)})$, is minimax optimal only guarantees the
estimation strategy is optimal in the sense that it both minimizes the
rate of convergence for the integrated mean squared error, a measure
of estimation performance, and that no other estimator will achieve a
faster rate of convergence for this performance metric. However, as we
can see from Remark \ref{remhnoptimal}, the choice of regularizing
parameters $\{h_{n,opt}\}_{n \geq 1}$ requires further investigation
by numerical or stochastic methods due to unknown constants that are
not directly estimable. For example, working with the approximately
optimal bandwidth choice in Remark \ref{remhnoptimal}, the constant
$C_{\Lambda}$ is proportional to the limit of $h_{n}^{2b + 1}\sum_{k =
  -\infty}^{\infty} \{\Lambda(h_nk)/\Psi(k)\}^2$, which could be
approximated by a finite series and a pilot regularizing sequence. On
the other hand, $C_{R}$ is essentially an asymptotically stabilized
bias. Usually, this is not observable, and, hence, a numerical method like
bootstrap or cross--validation is required to estimate it. In
addition, and more generally, the optimal bandwidth depends on the
unknown smoothness index $s$ of the function space
$\RR_{s}$. Estimation of this quantity is very difficult and likely
not even possible without harsh and confining assumptions. However,
an educated guess would lead to the optimal bandwidth choice
corresponding with the largest possible function class $\RR_{s}$. This
means choosing $s$ to be as small possible. Unfortunately, the
resulting methodology is still arbitrary.


\subsection{Smooth bootstrap of residuals}
\label{boot}

Computational approaches for automated and data--driven bandwidth
selection methods have been well--studied in the literature for many
nonparametric function estimators. The approaches generally focus on
estimating the IMSE of the estimator using either a cross--validation
or bootstrap approach, which can then be minimized with respect to the
choice of bandwidth in an exact or approximate way. Cao (1993) studies
two methods for selecting a bandwidth in a kernel density estimator
using a smooth bootstrap of their univariate data. More recently,
Neumeyer (2009) has proven the general validity of a smooth bootstrap
process of the model residuals from a nonparametric regression. Due to
its simplicity, we will introduce a similar smooth bootstrap process
that admits a consistent estimate of the IMSE of $\htheta$, which
requires mirroring the restrictions given by Theorem
\ref{thmFhatExpan} on model \eqref{modeleq} in the bootstrap
scheme. Throughout this section, we will describe the stochastic
properties of our random quantities using $P^*$--outer measure, which,
for a single bootstrap response $Y^*$, reduces to the conditional
probability function
\begin{equation*}
P_x^*(Y^* \leq t) = P_x(Y^* \leq t \,|\, \Dset)
= P_x(\ve^* \leq t - [K\htheta](x) \,|\, \Dset)
\end{equation*}
given the original sample of data $\Dset =
\{(x_{-n},Y_{-n}),\ldots,(x_{n},Y_{n})\}$. Here $\ve^*$ is a smooth
bootstrap model residual, which we construct as follows.

Let us begin with examining the requirements imposed by Theorem
\ref{thmFhatExpan} on model \eqref{modeleq}. We need to ensure our
smooth bootstrap model residual $\ve^*$ satisfies having a mean equal
to zero, independence, a finite moment of order $\kappa > 2 + 1/(s +
b)$ and a common distribution function $F_n^*$ that admits a bounded
Lebesgue density function $f_n^*$ that is H\"older continuous. The
first requirement is satisfied merely by centering our original model
residuals:
\begin{equation*}
\tilde{\ve}_j = \hve_j - \frac{1}{2n + 1} \sum_{l = -n}^{n} \hve_l,
 \qquad j = -n,\ldots,n.
\end{equation*}
Turning our attention to the next constraint, we can see that
conditioning on the original sample $\Dset$ and selecting from
$\tilde{\ve}_{-n},\ldots,\tilde{\ve}_{n}$ completely at random and
with replacement satisfies independence, in the sense of $P^*$--outer
measure. However, the remaining assumptions are not satisfied because
resampling in this way results in the bootstrap model residuals
$\tilde{\ve}_{j}^*$ having a discrete distribution.

To fulfill the last requirements imposed on model \eqref{modeleq}, we
will contaminate the randomly selected centered model residual
$\tilde{\ve}_{j}^*$ by an independent, centered random variable $U_j$
that has a finite moment of order $\kappa > 2 + 1/(s + b)$ and common
distribution function characterized by a bounded Lebesgue density
function $w$. Hence, we construct our smooth bootstrap model residuals
$\ve_{-n}^* = \tilde{\ve}_{-n}^* + c_nU_{-n},\ldots,\ve_{n}^* =
\tilde{\ve}_{n}^* + c_nU_n$. Here the sequence $\{c_n\}_{n \geq 1}$ is
a scaling sequence similar to a bandwidth for kernel density
estimation. Consequently, $\ve_{j}^*$ has the common distribution
function
\begin{equation} \label{Fnstar}
F_{n}^*(t) = P^*(\ve_{j}^* \leq t) = \frac{1}{(2n + 1)c_n}
 \sum_{j = -n}^{n} \int_{-\infty}^{t}
 w\bigg(\frac{u - \tilde{\ve}_j}{c_n}\bigg)\,du,
 \qquad t \in \R,
\end{equation}
and density function
\begin{equation*}
f_{n}^*(t) = \frac{1}{(2n + 1)c_n} \sum_{j = -n}^{n}
 w\bigg(\frac{t - \tilde{\ve}_j}{c_n}\bigg),
 \qquad t \in \R.
\end{equation*}
We can see that $F_n^*$ is a smooth estimator of $F$ based on a kernel
density estimator $f_n^*$ of the original error density $f$. Hence,
the remaining requirement imposed by Theorem \ref{thmFhatExpan} on
$F$ can be mirrored in the bootstrap process by choice of $w$, i.e.\
we can choose $w$ to be H\"older continuous with the desired
exponent. Using model \eqref{modeleq}, we obtain our bootstrap sample
$(x_{-n},Y_{-n}^*),\ldots,(x_{n},Y_{n}^*)$, where
\begin{equation*}
Y_{j}^* = \big[K\htheta\big](x_{j}) + \ve_j^*,
 \qquad j = -n,\ldots,n.
\end{equation*}

Define $\htheta^*$ as in \eqref{thetahat} but now with $Y_j^*$
replacing $Y_j$ and a regularizing sequence $\{g_n\}_{n \geq 1}$
replacing the regularizing sequence $\{h_n\}_{n \geq 1}$, which is
also chosen to satisfy \eqref{bw}. Choosing the scaling sequence
$\{c_n\}_{n \geq 1}$  such that $c_n = O(n^{-\alpha})$, for some   $0
< \alpha < 1/2 + 1/\kappa < 1$, results in the  bootstrap indirect
regression estimator $\htheta^*$ satisfying similar properties as
$\htheta$ given in Theorem \nolinebreak
\ref{thmthetahatuniformrate}. We summarize these results in
Proposition \ref{prophthetabootuniformrate} in Section
\ref{details}.

A particularly important use of bootstrapping in practice is to find
suitable quantiles for test statistics. In the case of residual-based
analysis, one is typically interested in functionals $T(F)$ of the
error distribution $F$ and wishes to test (say) $H_0\,:\,T(F) = 0$
versus $H_a\,:\,T(F) \neq 0$. Here $F$ is unknown, which requires
investigating the estimate $T(\hFemp)$ that is based on model
residuals. Neumeyer (2009) uses a smooth bootstrapping of
residuals obtained from nonparametric smoothing in a direct regression
model to estimate quantiles of $T_n = T(\hFemp)$ using a bootstrap
version of this quantity. We therefore expect analogous results from
Neumeyer (2009) to hold in the present context.

In the following, we work with residuals constructed from this
bootstrap data:
\begin{equation*}
\hve_j^* = Y_j^* - \big[K\htheta^*\big](x_j),
 \qquad j = -n,\ldots,n.
\end{equation*}
The following result is the analog of Theorem \ref{thmFhatExpan} for
the empirical distribution function of these residuals. The proof of
this result follows along the same lines as the proof of Theorem
\ref{thmFhatExpan} and its supporting results (see Section
\ref{details}). These have been omitted for brevity.

\begin{theorem} \label{thmFhatbootExpan}
Assume the density function $w$ is H\"older continuous with exponent
$0 < \gamma \leq 1$. Let the assumptions of Proposition
\ref{prophthetabootuniformrate} be satisfied, with $s > \max\{(2b +
1)/(2\gamma),\,3/2\}$. Then
\begin{equation*}
\sup_{t \in \R} \bigg| \frac{1}{2n + 1} \sum_{j = -n}^{n}
 \Big\{ \1\big[\hve_j^* \leq t\big] - \1\big[\ve_j^* \leq t\big]
 - \ve_j^*f_n^*(t) \Big\} \bigg| = o_{P^*}(n^{-1/2}).
\end{equation*}
\end{theorem}

Note, this result always includes the optimal bandwidth choice $c_n =
O(n^{-1/5})$ for density estimation. This fact in combination with the
results of Proposition \ref{propGausCovTerms} yield the following
analog of Corollary \ref{corFhatGaussian}:

\begin{corollary} \label{corFhatbootGaussian}
Let the assumptions of Theorem \ref{thmFhatbootExpan} be
satisfied. If, additionally, both of the densities $f$ and $w$ are
H\"older continuous with exponent $2/3 < \gamma \leq 1$, the scaling
sequence $\{c_n\}_{n \geq 1}$ satisfies $c_n = O(n^{-1/5})$ and $s >
(1 + \gamma)(2b + 1)/(3\gamma - 2)$, then the process
\begin{equation*}
(2n + 1)^{-1/2} \sum_{j = -n}^{n}
 \Big\{\1\big[\hve_j^* \leq t\big] - F_n^*(t)\Big\}
 = (2n + 1)^{-1/2}\sum_{j = -n}^{n}
 \Big\{\1\big[\ve_j^* \leq t\big] - F_n^*(t) + \ve_j^*f_n^*(t)\Big\}
 + o_{P^*}(1),
\end{equation*}
$t \in \R$, weakly converges, conditionally on the sample
$(x_{-n},\,Y_{-n}),\ldots,(x_{n},\,Y_{n})$, to a mean zero Gaussian
process $\{Z^*(t)\,:\,t \in \R\}$, with covariance function, for $u,v
\in \R$,
\begin{align*}
\Sigma^*(u,v) &= F_n^*\big(\min\{u,\,v\}\big) - F_n^*(u)F_n^*(v)
 + f_n^*(u)E^*\big[\ve^*\1[\ve^* \leq v]\big] \\
&\quad + f_n^*(v)E^*\big[\ve^*\1[\ve^* \leq u]\big]
 + \sigma^{2,*}f_n^*(u)f_n^*(v),
\end{align*}
where $\ve^*$ is a generic random variable with distribution function
$F_n^*$ and $\sigma^{2,*} = E^*[(\ve^*)^2]$. Additionally, we have
\begin{equation*}
\sup_{u,v \in \R} \Big| \Sigma^*(u,v) - \Sigma(u,v) \Big| = \op,
\end{equation*}
where $\Sigma$ is given in Corollary \ref{corFhatGaussian}.
\end{corollary}

Following the observations on pages 207-209 in Neumeyer (2009), we
can immediately obtain valid smooth bootstrap approximations of
quantiles for test statistics that are constructed from continuous
functionals of $F$ by combining this fact with the continuous sample
paths of Gaussian processes and the continuous mapping theorem. We
conclude this section with the following remark:
\begin{remark} \label{smoothbootvalidity}
Both the original residual--based process,
\begin{equation*}
(2n + 1)^{-1/2} \sum_{j = -n}^{n} \Big\{
 \1\big[\hve_j \leq t\big] - F(t) \Big\},
\end{equation*}
and its smooth bootstrap analogue,
\begin{equation*}
(2n + 1)^{-1/2} \sum_{j = -n}^{n} \Big\{
 \1\big[\hve_j^* \leq t\big] - F_n^*(t) \Big\},
\end{equation*}
have the same limiting distribution when the conditions of Corollary
\ref{corFhatbootGaussian} are satisfied. This limiting distribution is
given by the Gaussian process described in Corollary
\ref{corFhatGaussian}, which has continuous sample paths. It then
follows for statistics $T_n$ and their smooth bootstrap version
$T_n^*$ obtained from continuous functionals of $F$ satisfy the
following consistency property. Define $q_{n,\alpha}^*$ by $P^*(T_n^*
\leq q_{n,\alpha}^*) = \alpha$. Combining the continuity of the
functional used to construct $T_n$ and $T_n^*$ and the continuous
sample paths of Gaussian processes with the continuous mapping
theorem, we obtain
\begin{equation*}
P\big(T_n \leq q_{n,\alpha}^*\big) = \alpha + o(1),
\end{equation*}
which characterizes the validity of the proposed smooth bootstrap of
the model residuals. Hence, the bootstrap described here can be
immediately used to approximate unknown quantiles of test statistics
obtained from continuous functionals of $F$.
\end{remark}


\subsection{Regularization parameter selection by bootstrap}
\label{boothn}

Now we turn our attention to a different choice of regularization
parameters that also approximately minimizes the IMSE of the
indirect regression estimator $\htheta$. For clarity, throughout this
section we will subscript the estimators $\htheta$ and $\htheta^*$ by
the regularization parameters used to form them, i.e.\ we
write $\htheta_{h_n}$ to indicate the regularizing sequence
$\{h_n\}_{n \geq 1}$ is used to form the estimator $\htheta$. The IMSE
of $\htheta_{h_n}$, which we want to minimize with respect to the
regularizing sequence $\{h_n\}_{n \geq 1}$, is given by
\begin{equation} \label{IMSE}
IMSE\big(\htheta_{h_n}\big) = \int_{-1/2}^{1/2} E\Big[
 \big\{\htheta_{h_n}(x) - \theta(x)\big\}^2\Big]\,dx,
\end{equation}
which can be viewed as an objective function with respect to the
mapping $h_n \mapsto IMSE(\htheta_{h_n})$.

Following Cao (1993), we will arbitrarily choose the original
regularizing sequence $\{h_n\}_{n \geq 1}$ according to Theorem
\ref{thmthetahatuniformrate} as a pilot sequence to form an initial
and consistent estimator $\htheta_{h_n}$. A practical choice for
$\{h_n\}_{n \geq 1}$ is the rule-of-thumb parameter sequence given
in Remark \ref{remhnoptimal}, where the unknown constants are
estimated and the smoothness index $s$ is chosen as small as
possible. However, it is crucial for our approach to admit an
asymptotically optimal choice of regularizing parameters that the
pilot sequence $\{h_n\}_{n \geq 1}$ is chosen such that $s > s_0$,
where $s_0$ is the largest possible (finite) smoothness index such
that $\theta \in \RR_{s_0}$.

Consider the IMSE objective but now for the bootstrap
data, where we instead have $\htheta_{h_n}$ for the unknown function
$\theta$ in \eqref{IMSE}. Hence, we have an analogous form of
\eqref{IMSE} in $P^*$--outer measure that can be approximated via
Monte Carlo simulation:
\begin{equation} \label{IMSEboot}
IMSE^*\big(\htheta^*_{g_n}\big) = \int_{-1/2}^{1/2} E^*\Big[
 \big\{\htheta^*_{g_n}(x) - \htheta_{h_n}(x)\}^2\Big]\,dx.
\end{equation}

Since both $\htheta_{h_n}$ and $\htheta^*_{g_n}$ satisfy the
projective representation \eqref{thetahat}, it follows for the
expected values on the far right--hand sides of \eqref{IMSE} and
\eqref{IMSEboot} to be averages taken with respect to the distribution
functions $F$ and $F_n^*$, respectively. We can then use standard
arguments to show
\begin{equation*}
E^*\bigg[\int_{-1/2}^{1/2}
 \big\{\htheta^*_{g_n}(x) - \htheta_{h_n}(x)\big\}^2\,dx \bigg]
 = E\bigg[\int_{-1/2}^{1/2}
 \big\{\htheta_{h_n}(x) - \theta(x)\big\}^2\,dx \bigg] + \op.
\end{equation*}
Hence, we obtain $IMSE^*(\htheta^*_{g_n}) = IMSE(\htheta_{h_n}) +
\op$. This implies \eqref{IMSEboot} is an effective predictor of
\eqref{IMSE}, which implies that we can use the mapping $g_n \mapsto
IMSE^*(\htheta^*_{g_n})$ as an objective criteria for finding an
optimal regularizing sequence. It follows that we can choose
$\{g_{n,opt}\}_{n \geq 1}$ such that
\begin{equation} \label{gnopt}
g_{n,opt} = \argmin_{g \in (0,\,\hbar]} E^*\bigg[
 \int_{-1/2}^{1/2} \big\{\htheta^*_{g}(x) - \htheta_{h_n}(x)\big\}^2\,dx
 \bigg],
\end{equation}
where $\hbar > 0$ is a constant chosen larger than the optimal
regularization parameter. Consequently, the resulting regularization
parameters $\{g_{n,opt}\}_{n \geq 1}$ can be viewed as objective
corrections to the subjective pilot regularization parameters
$\{h_n\}_{n \geq 1}$.

Recall the Fourier frequency smoothing kernel $\Lambda$ used in the
deconvolution estimators $\htheta_{h_n}$ and $\htheta^*_{g_n}$. It is
easy to see that restricting the choice of $\Lambda$, and, hence,
restricting the choice of the resulting deconvolution smoothing
kernel from \eqref{thetahat}, leads to unique minimizers for each of
\eqref{IMSE} and \eqref{IMSEboot}, respectively. For example, choosing
$\Lambda$ as an indicator function, e.g.\ working with spectral
cut-off estimators, leads to the deconvolution smoothing kernel in
\eqref{thetahat} to be a smooth function with infinitely many
derivatives. The desired consistency property between the smooth
bootstrap selected optimal regularizing sequence $\{g_{n,opt}\}_{n
  \geq 1}$ from \eqref{gnopt} that minimizes \eqref{IMSEboot} and the
desired optimal regularizing sequence $\{h_{n,opt}\}_{n \geq 1}$ that
minimizes \eqref{IMSE} then follows from the consistency of
$IMSE^*(\htheta^*_{g_{n,opt}})$ for the desired
$IMSE(\htheta_{h_{n,opt}})$. We summarize these observations in the
following remark.

\begin{remark} \label{bootstrapbandwidth}
Let $s_0$ be the largest finite $s$ for which $\theta \in
\RR_{s}$. Choose the pilot regularizing sequence $\{h_n\}_{n \geq 1}$
according to the rule-of-thumb given in Remark \ref{remhnoptimal} with
$s > s_0$. We can restrict our choice of smoothing kernel by its
Fourier transform $\Lambda$ that allows for $\{h_{n,opt}\}_{n \geq 1}$
to be the unique minimizer of \eqref{IMSE}. Now choose
$\{g_{n,opt}\}_{n \geq 1}$ according to \eqref{gnopt}. Since
$IMSE^*(\htheta^*_{g_{n,opt}})$ is consistent for
$IMSE(\htheta_{h_{n,opt}})$, we have the desired $g_{n,opt} =
h_{n,opt} + \op$.
\end{remark}


\section{Finite sample properties}
\label{sims}

We conclude this article with a small numerical study of the previous
results, and we investigate the effectiveness of our smooth bootstrap
methodology for selecting a regularization parameter. In the following
simulations, we chose two regression functions $\theta_1$ and
$\theta_2$, where
\begin{equation*}
\theta_1(x) = 3\exp\big(-20x^2\big)
 \quad\text{and}\quad
 \theta_2(x) = 1 + 3 \cos\big( 3 \pi x / 4 \big)
 - 4 \cos^2( 3 \pi x \big),
 \quad x \in [-1/2,\,1/2].
\end{equation*}
The distortion function $\psi$ is taken as the Laplace density with a
mean of zero and a scale of $1/10$ that has been restricted to the
interval $[-1/2,\,1/2]$ as in Example \ref{expsi}, which also
satisfies Assumption \ref{assumpPsi} for the choice $b = 2$. The fixed
covariates are taken as $x_j = j/(2n + 1)$, which is asymptotically
equivalent to $j/(2n)$. This choice allows us to use the fast Fourier
transform algorithm for estimating the functions $\theta_1$ and
$\theta_2$. Finally, we consider two cases for the model errors:
normally distributed errors, with mean zero and scale $2/3$, and
$t$--distributed errors, with four degrees freedom and scaled
to $2/3$ as well. Our simulations consider samples of sizes $51$,
$101$, $201$ and $301$, i.e.\ $n$ is taken as $25$, $50$, $100$ and
$150$.

We work with the smoothing kernel that has Fourier coefficients
satisfying
\begin{equation*}
\Lambda(k) = \begin{cases}
  1, & \text{if } |k| \leq 7, \\
  (|k|/7)^{-6}, & \text{if } 7 < |k| \leq n,\\
  0, & \text{otherwise},
\end{cases}
\end{equation*}
which leads to considering function spaces $\RR_{s}$, with $5/2 < s <
7/2$. In order to select an appropriate regularization parameter for
the indirect regression function estimators, we work with the pilot
sequences $h_{n,1} = 5(2n + 1)^{-1/11}\log^{1/11}(2n + 1)$, which
corresponds with the choice $s = 3$ in \eqref{bw}, for estimating
$\theta_1$ and $h_{n,2} = 2.5(2n + 1)^{-1/11}\log^{1/11}(2n + 1)$ for
estimating $\theta_2$.

To create the smooth bootstrap of the residuals we have used standard
normally distributed contaminates $U_j$ and Silverman's rule
for selecting a bandwidth in kernel density estimation, i.e.\ we take
the scaling sequence $c_n = 1.06\hat{\sigma}(2n + 1)^{-1/5}$, where
$\hat{\sigma}$ is the estimated standard deviation of the model
residuals obtained using the pilot regularizing sequence. Using 200
smooth bootstrap replications to construct a suitable approximation of
the IMSE of the estimates of each of $\theta_1$ and $\theta_2$, we
take 100 equally spaced candidate regularization parameters in an
interval $[l_n,\,u_n]$, where $l_n = (2n + 1)^{-1/10}$, which results
in undersmoothed estimators, and $u_n = 10(2n +
1)^{-1/12}\log^{1/12}(2n + 1)$, which results in oversmoothed
estimators. Following the discussion in Section \ref{boothn}, we choose
the optimal regularization parameter $g_{n, opt}$ as the grid point
that minimizes this approximate IMSE, which we then use to construct
the resulting function estimators of $\theta_1$ and $\theta_2$.

\begin{figure}
\centering
\includegraphics[width=0.3\textwidth]{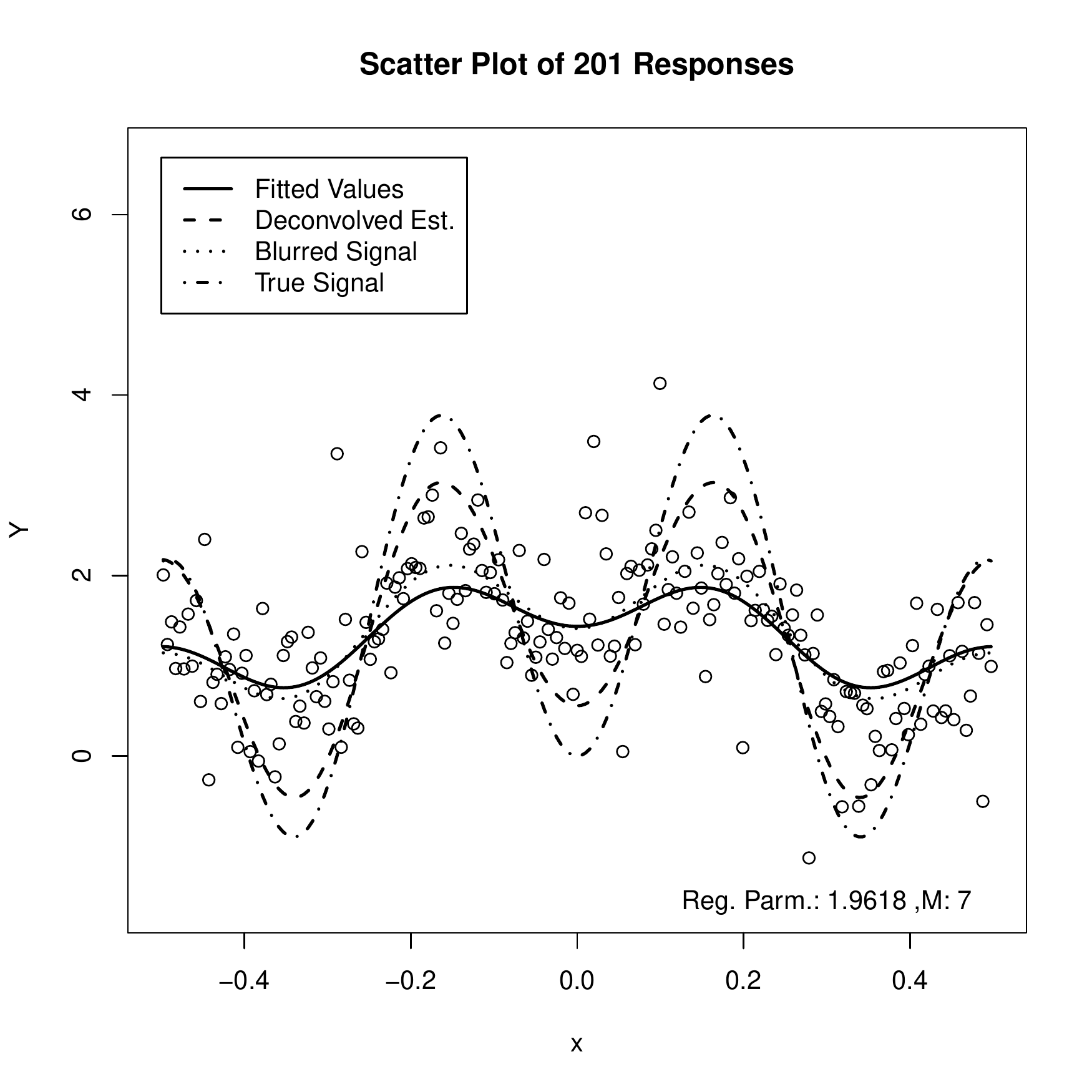}
\includegraphics[width=0.3\textwidth]{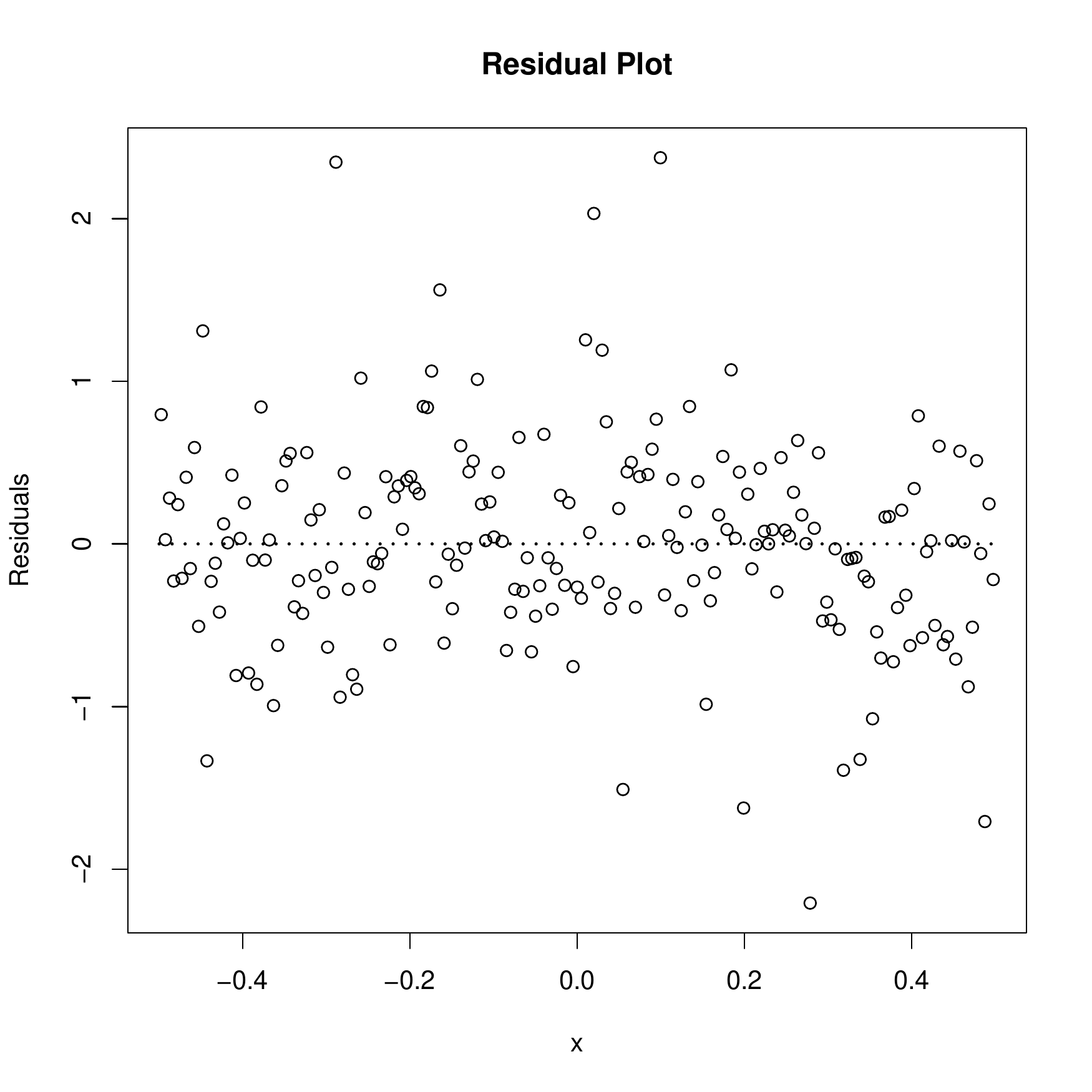}
\includegraphics[width=0.3\textwidth]{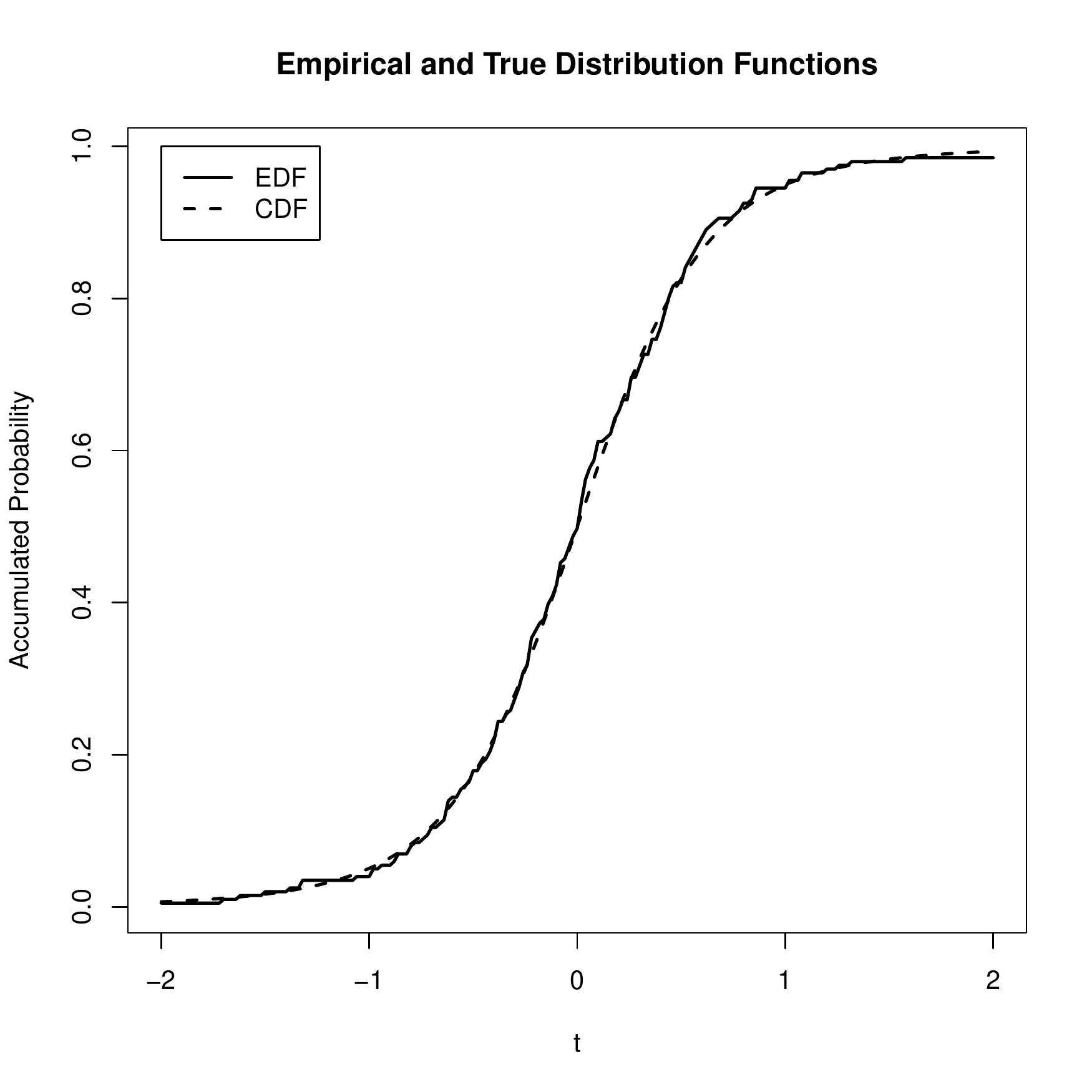}
\vspace{1ex}
\caption{{\it From left to right: A scatter plot of the data
  overlaid with the fitted blurred regression (solid),
  estimated regression (dashed), true blurred regression function
  (dotted) and the true regression function (dot-dashed);
  A scatter plot of the model residuals overlaid with a line at zero;
  A plot of the residual-based empirical distribution function (solid)
  overlaid with the true error distribution function
  (dashed).}}
\label{figure1}
\end{figure}

The assumptions of Theorem \ref{thmFhatExpan} are satisfied for the
choices made above. Figure \ref{figure1} displays the results of our
indirect regression estimator for a typical data set obtained from the
indirect regression $\theta_2$ and $t$--distributed errors based on a
sample size of 201. The scatter plot of the data shows the function
estimators $\htheta$ and $K\htheta$ work well in respectively
estimating each of $\theta_2$ and $K\theta_2$. We can plainly see that
the indirect regression estimator, constructed with the proposed
data-driven regularization methodology, is explaining the data very
well, which follows from the appearance of completely random scatter
in the plot of the residuals. The plot of the distribution functions
shows the empirical distribution function of the residuals $\hFemp$
matches very closely to the true error distribution function $F$ as
expected.

\begin{table}
\centering\small
\begin{tabular}{|r|c c c c c|}
\hline
\diagbox{$n$}{$t$} & -2 & -1 & 0 & 1 & 2 \\
\hline
$51$ & $0.000$ $(0.001)$ & $-0.017$ $(0.049)$ & $0.015$ $(0.086)$
  & $0.016$ $(0.045)$ & $0.001$ $(0.001)$ \\
     & $-0.004$ $(0.001)$ & $-0.076$ $(0.041)$ & $0.004$ $(0.085)$
  & $0.076$ $(0.040)$ & $0.005$ $(0.001)$ \\
$101$ & $-0.001$ $(0.001)$ & $-0.040$ $(0.046)$ & $-0.008$ $(0.090)$
  & $0.048$ $(0.047)$ & $0.003$ $(0.001)$ \\
      & $-0.003$ $(0.001)$ & $-0.066$ $(0.044)$ & $-0.017$ $(0.092)$
  & $0.076$ $(0.045)$ & $0.004$ $(0.001)$ \\
$201$ & $0.000$ $(0.001)$ & $-0.033$ $(0.043)$ & $0.017$ $(0.086)$
  & $0.048$ $(0.047)$ & $0.003$ $(0.001)$ \\
      & $-0.005$ $(0.001)$ & $-0.055$ $(0.044)$ & $-0.006$ $(0.081)$
  & $0.062$ $(0.042)$ & $0.004$ $(0.001)$ \\
$301$ & $-0.002$ $(0.001)$ & $-0.027$ $(0.041)$ & $-0.008$ $(0.083)$
  & $0.038$ $(0.044)$ & $0.004$ $(0.001)$ \\
      & $-0.004$ $(0.001)$ & $-0.056$ $(0.044)$ & $-0.015$ $(0.083)$
  & $0.061$ $(0.046)$ & $0.004$ $(0.001)$ \\
\hline
\end{tabular}
\vspace{1ex}
\caption{{\it Simulated asymptotic bias and variance (in parentheses)
    of $(2n + 1)^{1/2}\{\hFemp(t) - F(t)\}$ at the points $-2$, $-1$,
    $0$, $1$ and $2$ for the case of normally distributed errors. The
    results from each regression $\theta_1$ and $\theta_2$ are given as
    rows within each sample size, with the first row corresponding to
    $\theta_1$ and the second to $\theta_2$.}}
\label{BiasVarhFTableNormerr}
\end{table}

\begin{table}
\centering
\begin{tabular}{|c|c c c c c|}
\hline
\diagbox{$n$}{$t$} & $-2$ & $-1$ & $0$ & $1$ & $2$ \\
\hline
$51$ & $0.001$ & $0.049$ & $0.087$ & $0.045$ & $0.001$ \\
     & $0.001$ & $0.047$ & $0.085$ & $0.046$ & $0.001$ \\
$101$ & $0.001$ & $0.047$ & $0.090$ & $0.049$ & $0.001$ \\
      & $0.001$ & $0.049$ & $0.092$ & $0.051$ & $0.001$ \\
$201$ & $0.001$ & $0.045$ & $0.086$ & $0.046$ & $0.001$ \\
      & $0.001$ & $0.047$ & $0.081$ & $0.046$ & $0.001$ \\
$301$ & $0.001$ & $0.042$ & $0.083$ & $0.045$ & $0.001$ \\
      & $0.001$ & $0.048$ & $0.084$ & $0.050$ & $0.001$ \\
\hline
$\infty$ & $0.001$ & $0.046$ & $0.091$ & $0.046$ & $0.001$ \\
\hline
\end{tabular}
\vspace{1ex}
\caption{{\it Asymptotic mean squared error of $(2n +
    1)^{1/2}\{\hFemp(t) - F(t)\}$ at the points $-2$, $-1$, $0$, $1$
    and $2$ for the case of normally distributed errors. The results
    from each regression $\theta_1$ and $\theta_2$ are given as rows
    within each sample size, with the first row corresponding to
    $\theta_1$ and the second to $\theta_2$.}}
\label{AMSEhFTableNormerr}
\end{table}

\begin{table}
\centering
\begin{tabular}{|c c c c|c|}
\hline
51 & 101 & 201 & 301 & $\infty$ \\
\hline
$0.196$ & $0.191$ & $0.190$ & $0.186$ & $0.188$ \\
$0.200$ & $0.204$ & $0.190$ & $0.198$ &  \\
\hline
\end{tabular}
\vspace{1ex}
\caption{{\it Asymptotic integrated mean squared error of $(2n +
    1)^{1/2}\{\hFemp - F\}$ by sample size for the case of normally
    distributed errors. The results from each regression $\theta_1$
    and $\theta_2$ are given as rows, with the first row corresponding
    to $\theta_1$ and the second to $\theta_2$.}}
\label{AIMSEhFTableNormerr}
\end{table}

\begin{table}
\centering\small
\begin{tabular}{|r|c c c c c|}
\hline
\diagbox{$n$}{$t$} & -2 & -1 & 0 & 1 & 2 \\
\hline
$51$ & $-0.011$ $(0.005)$ & $-0.013$ $(0.037)$ & $0.025$ $(0.129)$
  & $-0.001$ $(0.040)$ & $0.007$ $(0.006)$ \\
     & $-0.013$ $(0.005)$ & $-0.038$ $(0.039)$ & $-0.002$ $(0.115)$
  & $0.057$ $(0.034)$ & $0.013$ $(0.004)$ \\
$101$ & $-0.007$ $(0.005)$ & $-0.021$ $(0.037)$ & $0.009$ $(0.139)$
  & $0.010$ $(0.037)$ & $0.008$ $(0.006)$ \\
      & $-0.015$ $(0.005)$ & $-0.038$ $(0.037)$ & $0.013$ $(0.136)$
  & $0.030$ $(0.041)$ & $0.014$ $(0.006)$ \\
$201$ & $-0.002$ $(0.006)$ & $-0.024$ $(0.033)$ & $-0.007$ $(0.153)$
  & $0.023$ $(0.038)$ & $0.006$ $(0.005)$ \\
      & $-0.011$ $(0.005)$ & $-0.032$ $(0.036)$ & $0.021$ $(0.153)$
  & $0.033$ $(0.038)$ & $0.010$ $(0.006)$ \\
$301$ & $-0.009$ $(0.005)$ & $-0.020$ $(0.036)$ & $0.005$ $(0.143)$
  & $0.007$ $(0.037)$ & $0.006$ $(0.006)$ \\
      & $-0.014$ $(0.006)$ & $-0.021$ $(0.037)$ & $-0.006$ $(0.150)$
  & $0.032$ $(0.034)$ & $0.010$ $(0.006)$ \\
\hline
\end{tabular}
\vspace{1ex}
\caption{{\it Simulated asymptotic bias and variance (in parentheses)
    of $(2n + 1)^{1/2}\{\hFemp(t) - F(t)\}$ at the points $-2$, $-1$,
    $0$, $1$ and $2$ for the case of $t$--distributed errors. The
    results from each regression $\theta_1$ and $\theta_2$ are given as
    rows within each sample size, with the first row corresponding to
    $\theta_1$ and the second to $\theta_2$.}}
\label{BiasVarhFTableTerr}
\end{table}

\begin{table}
\centering
\begin{tabular}{|c|c c c c c|}
\hline
\diagbox{$n$}{$t$} & $-2$ & $-1$ & $0$ & $1$ & $2$ \\
\hline
$51$ & $0.005$ & $0.037$ & $0.130$ & $0.040$ & $0.006$ \\
     & $0.005$ & $0.040$ & $0.115$ & $0.037$ & $0.005$ \\
$101$ & $0.005$ & $0.037$ & $0.139$ & $0.037$ & $0.006$ \\
      & $0.005$ & $0.039$ & $0.136$ & $0.042$ & $0.006$ \\
$201$ & $0.006$ & $0.034$ & $0.153$ & $0.038$ & $0.005$ \\
      & $0.006$ & $0.037$ & $0.153$ & $0.039$ & $0.006$ \\
$301$ & $0.005$ & $0.036$ & $0.143$ & $0.037$ & $0.006$ \\
      & $0.006$ & $0.037$ & $0.151$ & $0.035$ & $0.006$ \\
\hline
$\infty$ & $0.006$ & $0.036$ & $0.156$ & $0.036$ & $0.006$ \\
\hline
\end{tabular}
\vspace{1ex}
\caption{{\it Asymptotic mean squared error of $(2n +
    1)^{1/2}\{\hFemp(t) - F(t)\}$ at the points $-2$, $-1$, $0$, $1$
    and $2$ for the case of $t$--distributed errors. The results
    from each regression $\theta_1$ and $\theta_2$ are given as rows
    within each sample size, with the first row corresponding to
    $\theta_1$ and the second to $\theta_2$.}}
\label{AMSEhFTableTerr}
\end{table}

\begin{table}
\centering
\begin{tabular}{|c c c c|c|}
\hline
51 & 101 & 201 & 301 & $\infty$ \\
\hline
$0.237$ & $0.229$ & $0.222$ & $0.218$ & $0.228$ \\
$0.218$ & $0.231$ & $0.228$ & $0.223$ &  \\
\hline
\end{tabular}
\vspace{1ex}
\caption{{\it Asymptotic integrated mean squared error of $(2n +
    1)^{1/2}\{\hFemp - F\}$ by sample size for the case of
    $t$--distributed errors. The results from each regression
    $\theta_1$ and $\theta_2$ are given as rows, with the first row
    corresponding to $\theta_1$ and the second to $\theta_2$.}}
\label{AIMSEhFTableTerr}
\end{table}

Turning our attention to the numerical summaries of the estimator
$\hFemp$, we can plainly see this estimator is performing
well. Beginning with the case of normally distributed errors, Table
\ref{BiasVarhFTableNormerr} shows the figures for the simulated asymptotic
biases and variances of $\hFemp$ at the points $-2$, $-1$, $0$, $1$
and $2$. The simulated asymptotic biases are calculated by computing
the simulated biases of $\hFemp$ and multiplying these by the
square--root of the corresponding sample size, and the simulated
asymptotic variance is similarly calculated but now we multiply by the
corresponding sample size. Inspecting Table \ref{BiasVarhFTableNormerr}, we
find the squared asymptotic bias of $\hFemp$ becomes negligible to the
asymptotic variance of $\hFemp$ at larger sample sizes, which is
expected. In Table \ref{AMSEhFTableNormerr}, we give the asymptotic mean
squared error (AMSE) of $\hFemp$, which is calculated by multiplying
the simulated mean squared error of $\hFemp$ by the corresponding
sample size. The figures corresponding to the sample size $\infty$ are
calculated using the results of Theorem \ref{thmFhatExpan}. Comparing
the results in Table \ref{AMSEhFTableNormerr}, we find the theoretical
prediction made in Theorem \ref{thmFhatExpan} concerning the
asymptotic pointwise precision of $\hFemp$ corresponds well with the
simulated results. Finally, turning our attention to Table
\ref{AIMSEhFTableNormerr}, we give the asymptotic integrated mean squared
error (AIMSE) of $\hFemp$, which is calculated similarly to the AMSE
of $\hFemp$ but now integrating with respect to $t$. These results
also confirm that $\hFemp$ performs well in estimating $F$ even at
the smaller sample sizes 51 and 101. A possible explanation for this
observation is the use of the smooth bootstrap methodology for
choosing the regularization parameter in the estimate
$\htheta$. Table \ref{BiasVarhFTableTerr}, Table
\ref{AMSEhFTableTerr} and Table \ref{AIMSEhFTableTerr} show the
related figures to Table \ref{BiasVarhFTableNormerr}, Table
\ref{AMSEhFTableNormerr} and Table \ref{AIMSEhFTableNormerr},
respectively, when the model errors are $t$--distributed, and the
results are analogous to the case of normally distributed errors.

\begin{figure}
\centering
\includegraphics[width=0.35\linewidth]{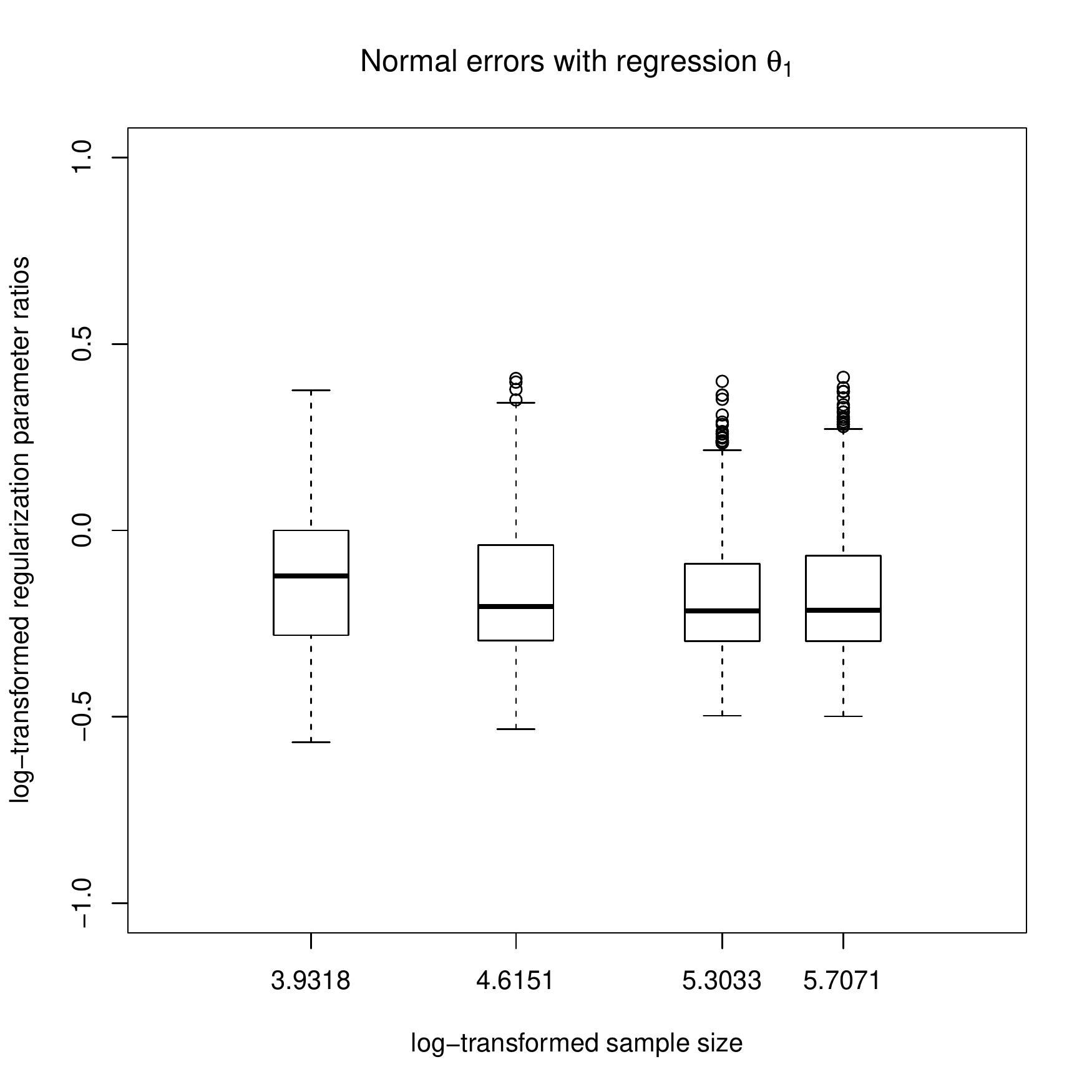}
\includegraphics[width=0.35\linewidth]{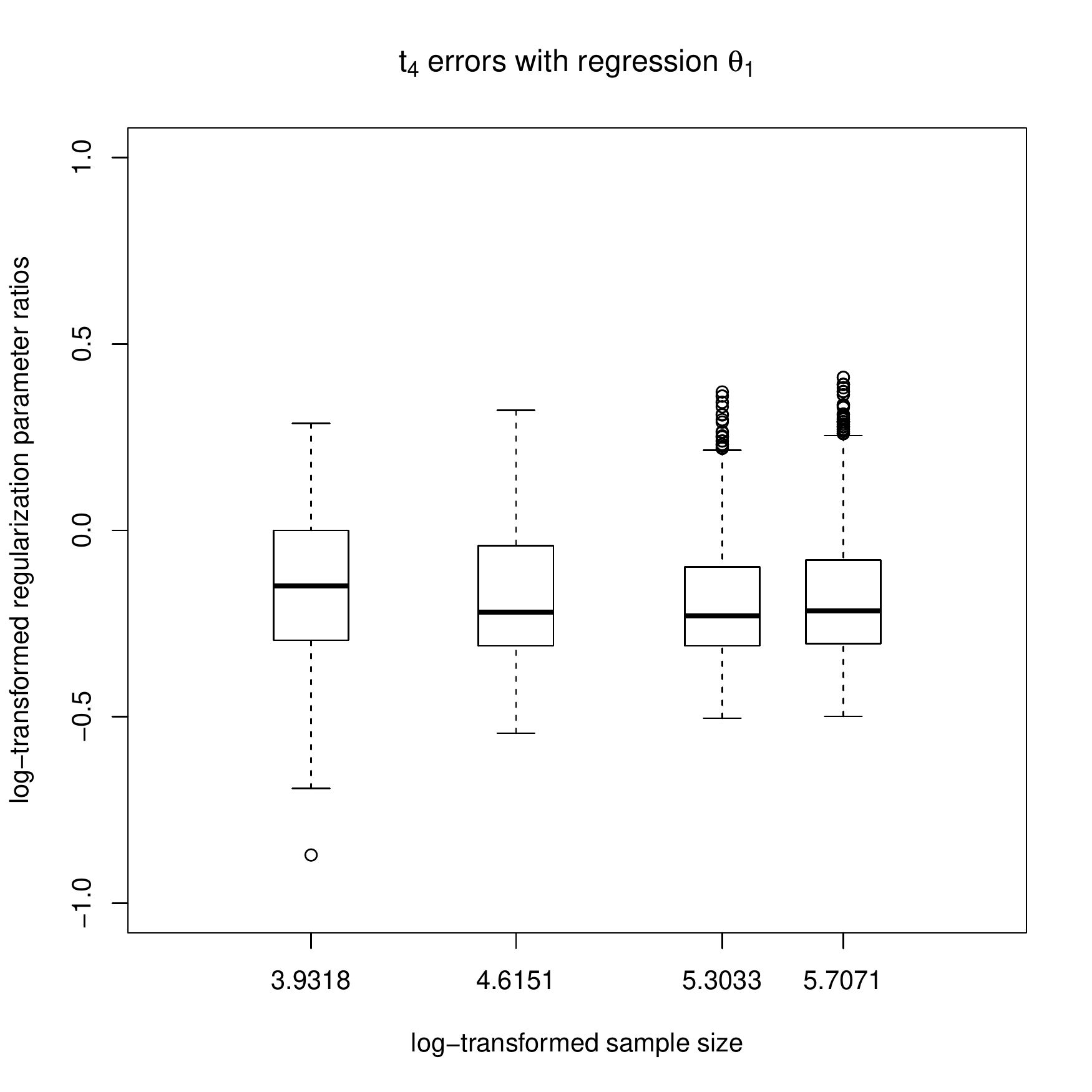}
\\
\includegraphics[width=0.35\linewidth]{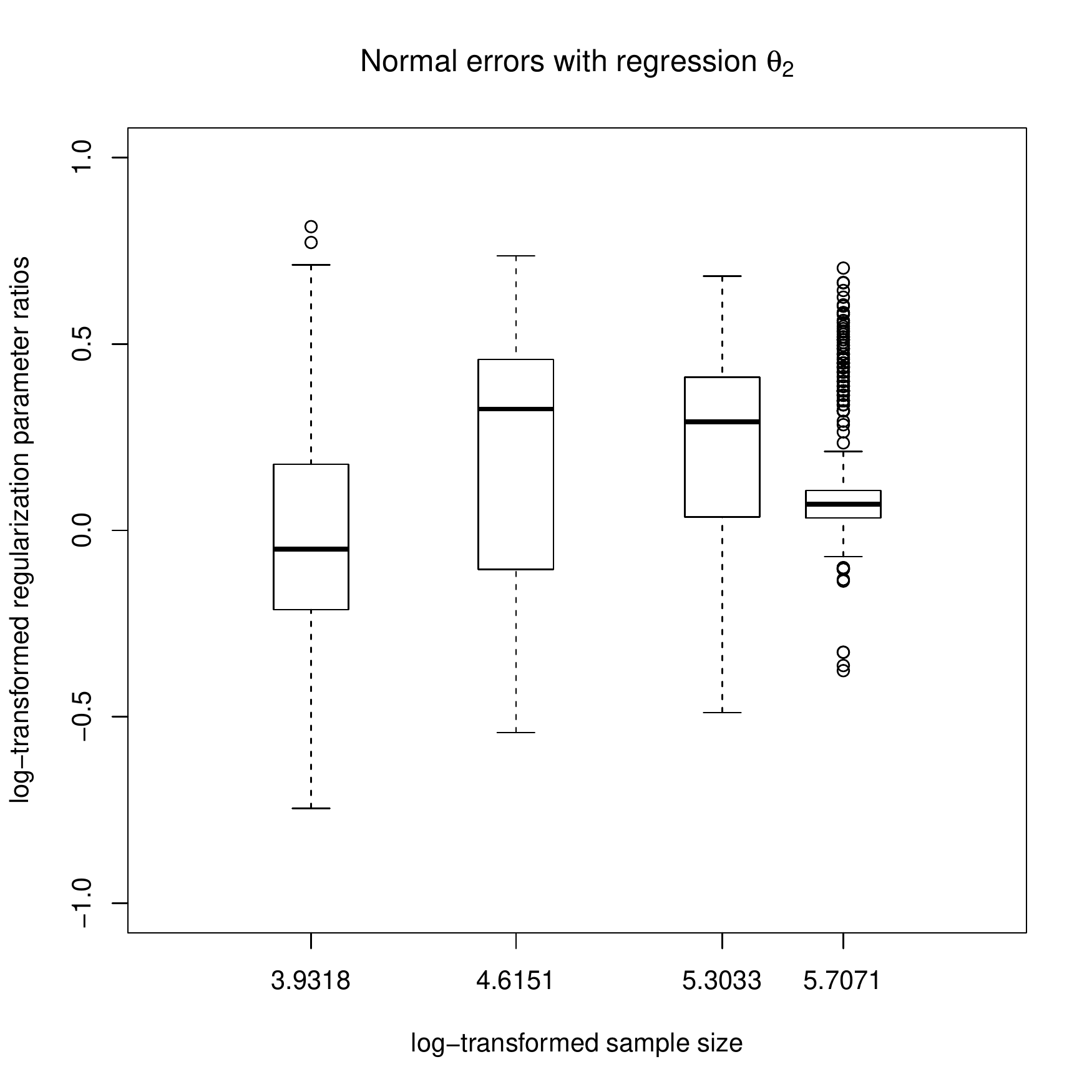}
\includegraphics[width=0.35\linewidth]{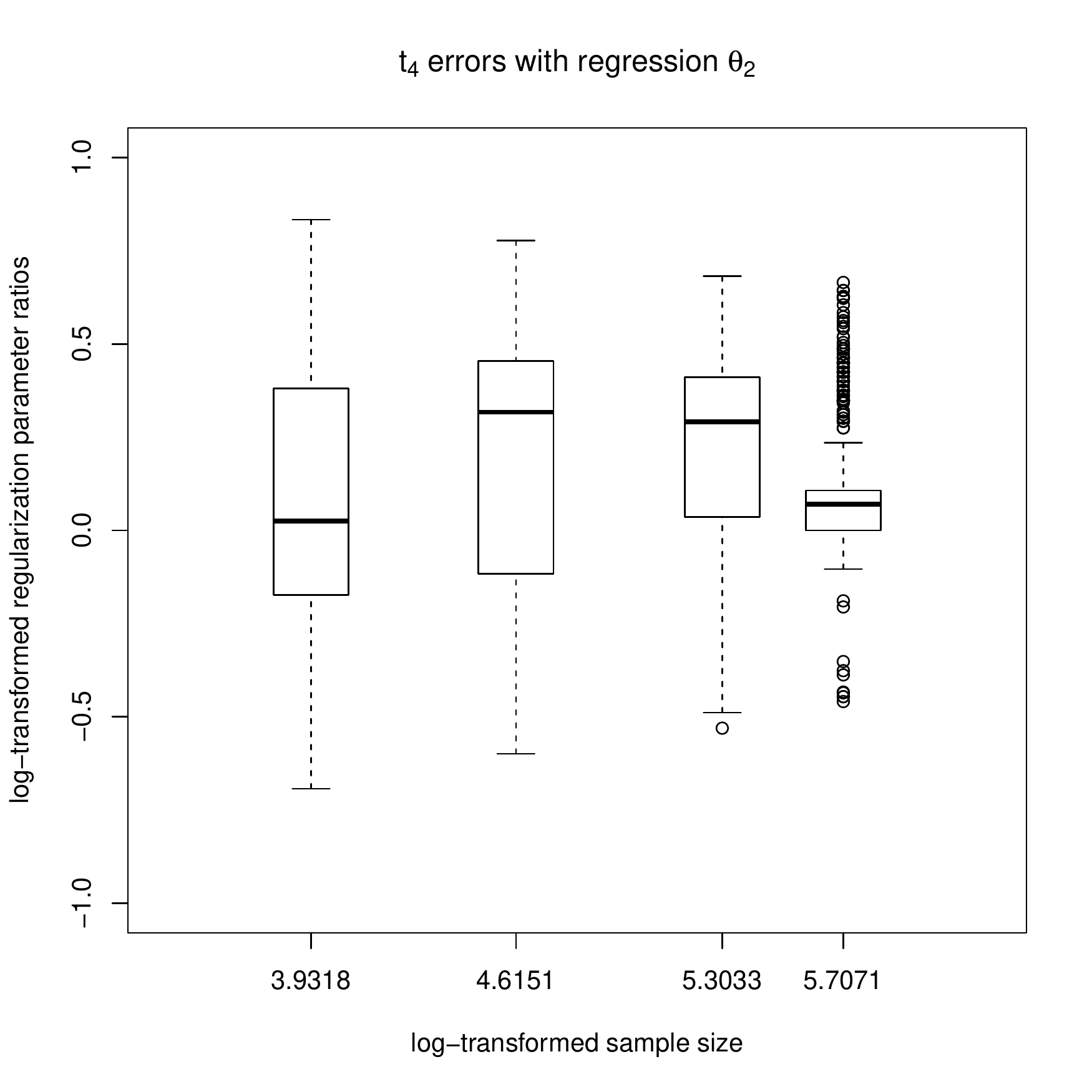}
\vspace{1ex}
\caption{{\it Boxplots of log-transformed ratios of regularization
    parameters (bootstrap--based selection to ISE--based selection) by
    log-transformed sample size. Plots on the left correspond to
    normally distributed errors and plots on the right correspond to
    $t$--distributed errors, while plots on the top correspond to the
    regression $\theta_1$ and plots on the bottom correspond to the
    regression $\theta_2$.}}
\label{bwdsfigure}
\end{figure}

\begin{table}
\centering
\begin{tabular}{|c|c c c c|}
\hline
Regularization & 51 & 101 & 201 & 301 \\
\hline
Bootstrap & $0.169$ & $0.093$ & $0.056$ & $0.040$ \\
          & $0.762$ & $0.601$ & $0.348$ & $0.249$ \\
Best & $0.132$ & $0.076$ & $0.049$ & $0.036$ \\
     & $0.587$ & $0.442$ & $0.283$ & $0.205$ \\
\hline
\end{tabular}
\vspace{1ex}
\caption{{\it Integrated mean squared error of the indirect regression
    estimator by sample size for each regression $\theta_1$ and
    $\theta_2$ in the case of normally distributed errors. Figures
    corresponding to `Bootstrap' are the IMSE estimates based on the
    proposed smooth bootstrap methodology for selecting the
    regularization parameter and the figures corresponding to `Best'
    are the IMSE estimates corresponding to selecting the
    regularization parameter by minimizing the ISE. The results for
    each regression $\theta_1$ and $\theta_2$ are given as rows within
    each regularization selection method, with the first row
    corresponding to $\theta_1$ and the second to $\theta_2$.}}
\label{IMSEhthetaTableNormerr}
\end{table}

\begin{table}
\centering
\begin{tabular}{|c|c c c c|}
\hline
Regularization & 51 & 101 & 201 & 301 \\
\hline
Bootstrap & $0.158$ & $0.094$ & $0.053$ & $0.040$ \\
          & $0.772$ & $0.584$ & $0.347$ & $0.250$ \\
Best & $0.123$ & $0.078$ & $0.046$ & $0.036$ \\
     & $0.598$ & $0.431$ & $0.283$ & $0.206$ \\
\hline
\end{tabular}
\vspace{1ex}
\caption{{\it Integrated mean squared error of the indirect regression
    estimator by sample size for each regression $\theta_1$ and
    $\theta_2$ in the case of $t$--distributed errors. Figures
    corresponding to `Bootstrap' are the IMSE estimates based on the
    proposed smooth bootstrap methodology for selecting the
    regularization parameter and the figures corresponding to `Best'
    are the IMSE estimates corresponding to selecting the
    regularization parameter by minimizing the ISE. The results for
    each regression $\theta_1$ and $\theta_2$ are given as rows within
    each regularization selection method, with the first row
    corresponding to $\theta_1$ and the second to $\theta_2$.}}
\label{IMSEhthetaTableTerr}
\end{table}

The results concerning our indirect regression estimator are
interesting. In addition to finding an asymptotically optimal
regularization parameter using the proposed bootstrap methodology, we
also conducted a similar grid search procedure choosing an optimal
regularization parameter that minimizes the integrated squared error
(ISE) between the indirect regression estimate and the regression
function for each case of $\theta_1$ and $\theta_2$. In general, this
methodology is not available in applications, but we expect it to
produce the best resulting indirect regression estimate with respect
to the IMSE of these estimates.

In Figure \ref{bwdsfigure} we give boxplots of the log-transformed
ratios of the optimal regularization parameter selected from the
proposed bootstrap methodology to the regularization parameter chosen
from the ISE methodology at each log-transformed sample size. At the
larger sample sizes, we can plainly see the boxes are beginning to
include 0, which we expect to continue as the sample size
increases. This confirms the conjecture of consistency between the two
regularization techniques mentioned in Remark
\ref{bootstrapbandwidth}. It appears that with increasing sample size
both the bootstrap selection methodology and the ISE selection
methodology choose similar regularizations for each of $\theta_1$ and
$\theta_2$ in both cases of normally distributed and $t$--distributed
errors.

We have also numerically measured the performance of the indirect
regression estimator by simulating the IMSE using both regularization
techniques for each regression $\theta_1$ and $\theta_2$ in both cases
of normally distributed errors and $t$--distributed errors. The
results are given in Table \ref{IMSEhthetaTableNormerr} for the case
of normally distributed errors and Table \ref{IMSEhthetaTableTerr}
for the case of $t$--distributed errors. We can plainly see that the
IMSE of the estimators using each regularization method are decreasing
to zero as the sample size increases, and the IMSE values between the
bootstrap--based method and the ISE--based method appear to be very
similar, even at the smaller sample sizes 51 and 101, which also
confirms the conjecture of consistency between the two regularization
techniques given in Remark \ref{bootstrapbandwidth}. In summary, we
find the residual--based empirical distribution function is performing
well as an estimator of the error distribution function, and the
proposed smooth bootstrap methodology for selecting the regularization
parameter used in the indirect regression estimate provides a useful
and convenient tool for precise indirect regression function
estimation.


\subsection{Example: comparison between regularization methods for
  spectral cut-off estimators}
\label{comparemethods}

Consider the special case of indirect regression estimates from the
so-called spectral cut-off method. This means we consider the simpler
spectral smoothing kernel
\begin{equation*}
\Lambda(k) = \1\big[-1 \leq k \leq 1\big],
 \qquad k \in \Zint.
\end{equation*}
Here one seeks a regularization that essentially decides how many
Fourier frequencies to include in the indirect regression estimator,
which follows from observing that \eqref{thetahat} evaluates $\Lambda$
at the product $h_n k$, where the regularizing parameter $h_n$ is
small. Cavalier and Golubev (2006) investigate a penalized estimator
of the mean integrated squared error of indirect regression estimators
obtained from the spectral cut-off method called a {\em risk hull};
see equation (1.9) on page 1656. These authors propose selecting a
regularization that minimizes this quantity and call this the risk
hull method. A penalty function is given but an approximate
specification is only provided for the case of normally distributed
data (see page 1661 of that article), where the authors write on page
1660 that ``this approximation is not good for small $k$'' and point
directly to an unspecified Monte Carlo strategy for computing their
penalty function.

In addition to choosing an appropriate penalty function, the risk hull
method also requires choosing a tuning parameter $\alpha$ that
influences the strength of the penalty. Here the authors suggest using
$\alpha = 1.1$ on page 1664, which we use as well. The main drawback
with the tuning parameter $\alpha$ is the risk hull method appears not
to work very well when alpha is chosen either too large or too small,
which seems to imply an optimal sequence $\{\alpha_n\}_{n \geq 1}$
should be used instead. The authors make a disappointing remark on
page 1659: ``We do not believe that there is a good general formula
for the optimal risk hull or for the penalty.'' Hence, without further
guidance we simply use the provided approximate penalty given on page
1661 (scaled by $1/10$). This is in contrast to the proposed bootstrap
selector that {\em objectively corrects} the arbitrary choice of pilot
regularization.

In this example, we simulated a comparison between the risk hull
method and the proposed bootstrap regularization selection method from
Section \ref{boothn} for both regressions $\theta_1$ and
$\theta_2$. The distortion function $\psi$ is specified in Section
\ref{sims} and the errors are again normally distributed with mean
zero and scale $2/3$. As before, we considered sample sizes 51, 101,
201 and 301. For the bootstrap selection method, we used the same
pilot sequences that were used in the previous simulations.

\begin{figure}
\centering
\includegraphics[width=0.35\linewidth]{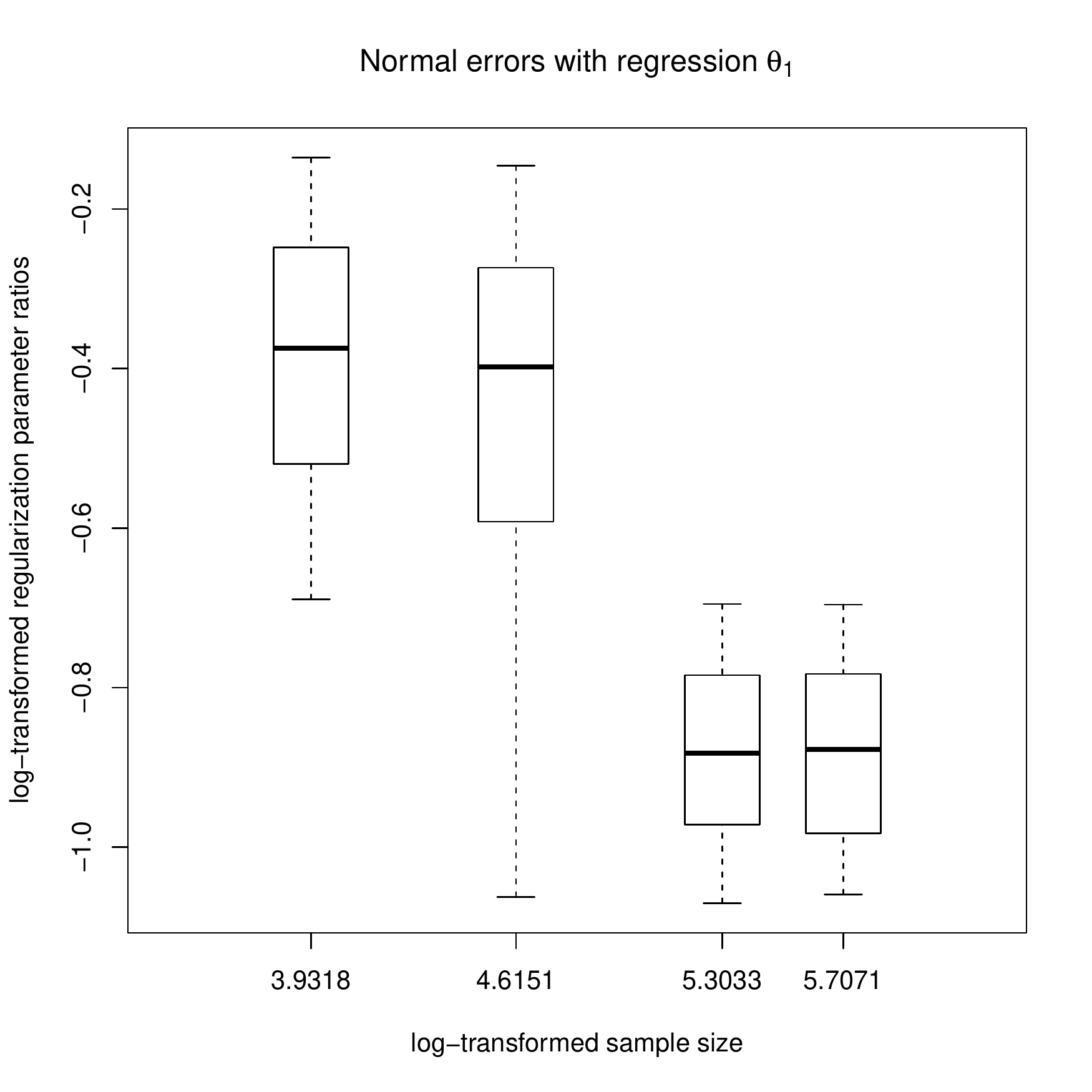}
\includegraphics[width=0.35\linewidth]{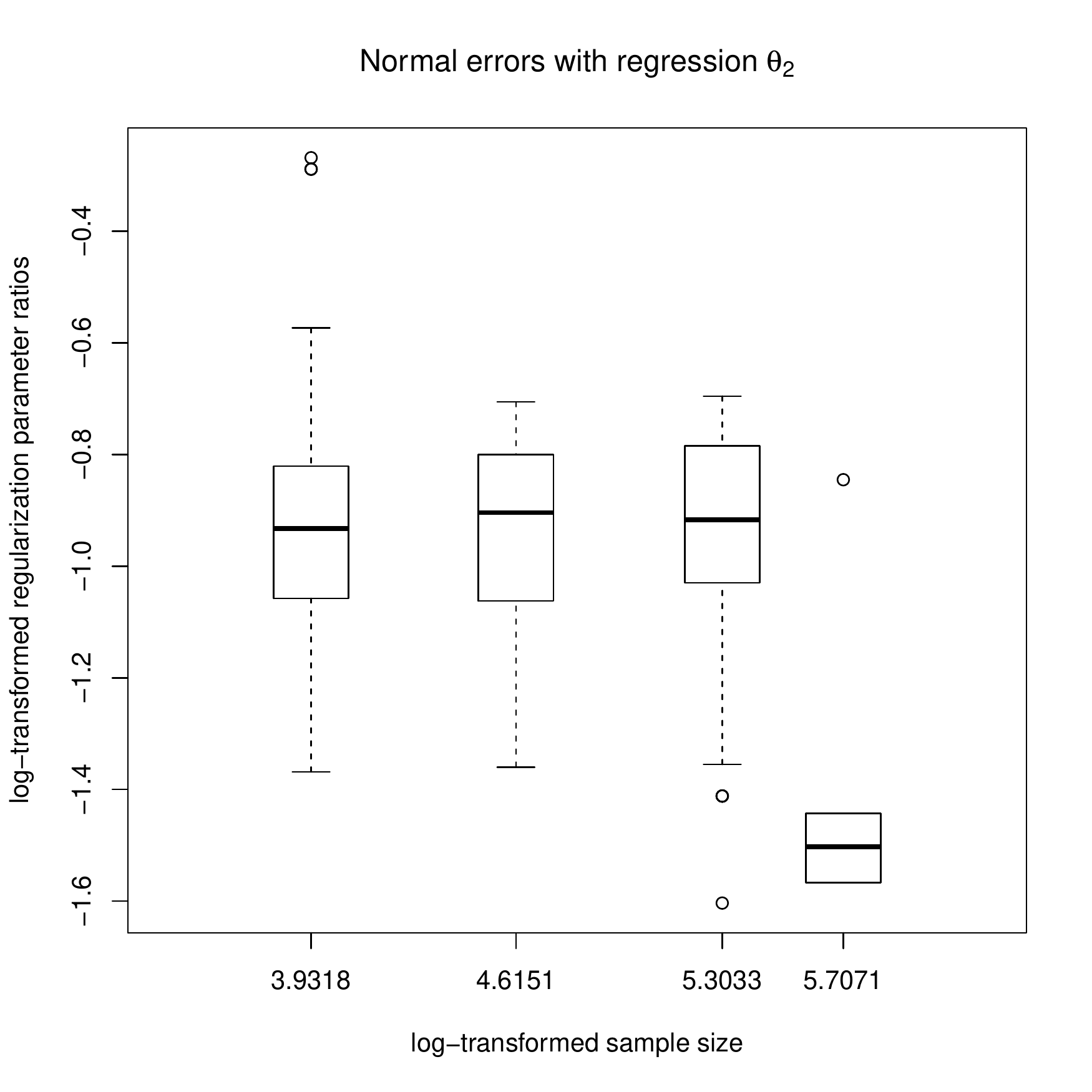}
\vspace{1ex}
\caption{{\it Boxplots of log-transformed ratios of regularization
    parameters (bootstrap--based selection to risk--hull--based
    selection) by log-transformed sample size. The boxplot on the left
    corresponds to the regression $\theta_1$ and the boxplot on the
    right to the regression $\theta_2$.}}
\label{BSRHMfigure}
\end{figure}

The results of our numerical study are summarized in the boxplots
displayed in Figure \ref{BSRHMfigure}. At smaller sample sizes 51 and
101, we can see that both approaches generally choose similar
regularizations, i.e.\ both procedures suggest similar spectral
cuts. The larger sample sizes 201 and 301, however, show the risk hull
method begins to favor regularizations that include more Fourier
frequencies than the bootstrap method. Consequently, for $\theta_1$
the simulated IMSE values are $0.321$, $0.303$, $0.290$ and $0.287$
using the risk hull method and $0.321$, $0.286$, $0.055$ and $0.041$
using the proposed bootstrap procedure, respectively for each sample size
51, 101, 201 and 301. Similarly, the simulated IMSE values for
$\theta_2$ are $1.978$, $1.957$, $1.946$ and $1.940$ using the risk
hull method and $0.762$, $0.611$, $0.541$ and $0.217$ using the
proposed bootstrap procedure. We can plainly see the proposed
bootstrap selection procedure compares favorably with the risk hull
method. Since the proposed bootstrap procedure is highly applicable,
we recommend practitioners to use it when considering data--driven
regularization selection procedures.

\medskip
\medskip

{\bf Acknowledgements}
We would like to thank the Referees for their careful reading and
helpful comments that improved the article. In particular,
one referee helped us to see how to improve and streamline our
approach as well as pointing out several important works in the
literature that had been previously unknown to us. This work has been
supported in part by the Collaborative Research Center ``Statistical
modeling of nonlinear dynamic processes'' (SFB 823, Projects C1 and
C4) of the German Research Foundation (DFG) and in part by the
Bundesministerium f\"ur Bildung und Forschung through the project
``MED4D: Dynamic medical imaging: Modeling and analysis of medical
data for improved diagnosis, supervision and drug development''.


\section{Technical details}
\label{details}


The estimator $\hR$ is biased only in the design points, which
asymptotically exhaust the interval $[-1/2,\,1/2]$ at the rate
$n^{-1}$. We arrive at the following result concerning the bias of
$\hR$:
\begin{lemma} \label{lemRhatOrder}
Let $\theta \in \RR_{s}$, with $s \geq 1$. Then
\begin{equation*}
\max_{k \in \Zint} \Big|E\big[\hR(k)\big] - R(k)\Big|
 = O\big(n^{-1}\big).
\end{equation*}
\end{lemma}
\begin{proof}
For any $s_1 \leq s_2$, we have the inclusion $\RR_{s_2} \subset
\RR_{s_1}$, and, therefore, we only need to prove the result for $s =
1$. For clarity, write $r = K\theta$. Without any loss of generality,
we can assume that $|r(0)| < \infty$. We can write
\begin{align} \label{EhR}
E\big[\hR(k)\big] &= \frac{1}{2n + 1} \sum_{j = -n}^{n}
 \bigg\{\int_{-\infty}^{\infty} y\,Q(dy\,|\,x_j)\bigg\} e^{-i2\pi k x_j}
 = \frac{1}{2n + 1} \sum_{j = -n}^{n} r(x_j)e^{-i2\pi k x_j}
 \\ \nonumber
&= \frac{1}{2n + 1} \sum_{j = 1}^{n} r(x_j)e^{-i2\pi k x_j}
 + \frac{r(0)}{2n + 1}
 + \frac{1}{2n + 1} \sum_{j = -n}^{-1} r(x_j)e^{-i2\pi k x_j}.
\end{align}
The second equality in \eqref{EhR} shows that $\hR$ is on the average
estimating the discrete Fourier transform of $r$ calculated on the
design points, which is expected.

We can relate the discrete Fourier transform of $r$ to its Fourier
coefficients $\{R(k)\}_{k \in \Zint}$ as follows. Partition the
interval $[-1/2,\,1/2]$ into
\begin{equation*}
\Bigg(\bigcup_{j = -n}^{-1}
 \bigg[\frac{2j - 1}{4n + 2},\,\frac{2j + 1}{4n + 2}\bigg)\Bigg)
\bigcup \bigg[-\frac{1}{4n + 2},\,\frac{1}{4n + 2}\bigg]
\bigcup \Bigg(\bigcup_{j = 1}^n
 \bigg(\frac{2j - 1}{4n + 2},\,\frac{2j + 1}{4n + 2}\bigg]\Bigg).
\end{equation*}
Since $x_j = j/(2n)$, we have $v/(2n + 1) + j/(2n + 1) = x_j + (v -
x_j)/(2n + 1)$. It follows that $E[\hR(k)] - R(k)$ is equal to
\begin{align*}
&\frac{1}{2n + 1} \sum_{j = 1}^{n} \bigg\{
 \int_{-1/2}^{1/2} \bigg\{r(x_j)
 - r\bigg(x_j + \frac{v - x_j}{2n + 1}\bigg)\bigg\}\,dv
 \bigg\}\exp(-i2\pi k x_j) \\
& + \frac{1}{2n + 1} \sum_{j = 1}^{n} \int_{-1/2}^{1/2}
 r\bigg(x_j + \frac{v - x_j}{2n + 1}\bigg)
 \bigg\{\exp(-i2\pi k x_j) - \exp\bigg(-i2\pi k
 \bigg(x_j + \frac{v - x_j}{2n + 1}\bigg)\bigg)\bigg\}
 \,dv \\
& + \frac{1}{2n + 1} \int_{-1/2}^{1/2}
 \bigg\{r(0) - r\bigg(\frac{v}{2n + 1}\bigg)\bigg\}
 \exp\bigg(-i2\pi k\frac{v}{2n + 1}\bigg)\,dv \\
& + \frac{r(0)}{2n + 1} \int_{-1/2}^{1/2}
 \bigg\{1 - \exp\bigg(-i2\pi k\frac{v}{2n + 1}\bigg)
 \bigg\}\,dv \\
& + \frac{1}{2n + 1} \sum_{j = -n}^{-1} \bigg\{ \int_{-1/2}^{1/2}
 \bigg\{ r(x_j) - r\bigg(x_j + \frac{v - x_j}{2n + 1}\bigg)
 \bigg\}\,dv\bigg\} \exp(-i2\pi k x_j) \\
& + \frac{1}{2n + 1} \sum_{j = -n}^{-1} \int_{-1/2}^{1/2} 
 r\bigg(x_j + \frac{v - x_j}{2n + 1}\bigg)
 \bigg\{\exp(-i2\pi k x_j) - \exp\bigg(-i2\pi k
 \bigg(x_j + \frac{v - x_j}{2n + 1}\bigg)\bigg)
 \bigg\}\,dv.
\end{align*}
We can see that $|E[\hR(k)] - R(k)|$ is bounded by
\begin{equation*}
 R_1(k) + R_2(k) + R_3(k) + R_4(k) + R_5(k) + O(n^{-1}),
\end{equation*}
where the error term $O(n^{-1})$ does not depend on $k$ and
\begin{equation*}
R_1(k) = \frac{1}{2n + 1} \sum_{j = 1}^{n}
 \int_{-1/2}^{1/2} \bigg|r(x_j)
 - r\bigg(x_j + \frac{v - x_j}{2n + 1}\bigg)\bigg|\,dv,
\end{equation*}
$R_2(k)$ is equal to
\begin{equation*}
\frac{1}{2n + 1} \sum_{j = 1}^{n} \bigg|
 \int_{-1/2}^{1/2} r\bigg(x_j + \frac{v - x_j}{2n + 1}\bigg)
 \bigg\{\exp\big(-i2\pi kx_j\big)
 - \exp\bigg( -i2\pi k
 \bigg(x_j + \frac{v - x_j}{2n + 1}\bigg) \bigg) \bigg\}
 \,dv\bigg|,
\end{equation*}
\begin{equation*}
R_3(k) = \frac{1}{2n + 1} \int_{-1/2}^{1/2} \bigg|
 r(0) - r\bigg(\frac{v}{2n + 1}\bigg)\bigg|\,dv,
\end{equation*}
\begin{equation*}
R_4(k) = \frac{1}{2n + 1} \sum_{j = -n}^{-1}
 \int_{-1/2}^{1/2} \bigg|r(x_j)
 - r\bigg(x_j + \frac{v - x_j}{2n + 1}\bigg)\bigg|\,dv,
\end{equation*}
and $R_5(k)$ is equal to
\begin{equation*}
\frac{1}{2n + 1} \sum_{j = -n}^{-1} \bigg|
 \int_{-1/2}^{1/2} r\bigg(x_j + \frac{v - x_j}{2n + 1}\bigg)
 \bigg\{\exp\big(-i2\pi kx_j\big)
 - \exp\bigg(-i2\pi k\bigg(x_j + \frac{v - x_j}{2n + 1}\bigg)
 \bigg)\bigg\}\,dv\bigg|.
\end{equation*}
Hence, the result follows, if we can show $\max_{k \in \Zint} R_i(k) =
O(n^{-1})$, for each $i = 1,\ldots,5$.

To continue, use Euler's formula to write
\begin{equation*}
\exp\big(-i2\pi kx_j\big) = \cos\big(2\pi k x_j\big)
 - i\sin\big(2\pi k x_j\big).
\end{equation*}
Since sine and cosine are each Lipschitz functions with constant equal
to $1$, it follows that
\begin{equation*}
\bigg| \exp\big(-i2\pi k x_j\big)
 - \exp\bigg(-i2\pi k\bigg(x_j + \frac{v - x_j}{2n + 1}\bigg)\bigg)
 \bigg|^2 \leq 2^3\pi^2k^2 \bigg(\frac{v - x_j}{2n + 1} \bigg)^2.
\end{equation*}
Therefore we have the bound
\begin{align*}
\bigg| \exp\big(-i2\pi k x_j\big)
 - \exp\bigg(-i2\pi k\bigg(x_j + \frac{v - x_j}{2n + 1}\bigg)\bigg)
 \bigg| \leq 2^{3/2}\pi|k| \frac{|v - x_j|}{2n + 1},
\end{align*}
which will be used throughout the proof.

Beginning with $R_1(k)$, it follows from both $\theta \in \RR_{1}$ and
the equivalence $r = K\theta$ that $r \in \RR_{1}$ as well. Using the
Fourier inversion formula, write
\begin{equation*}
\bigg|r\big(x_j\big) - r\bigg(x_j + \frac{v - x_j}{2n + 1}\bigg)\bigg|
 \leq \frac{2^{3/2}\pi}{2n + 1} \big|v - x_j\big|
 \sum_{k = -\infty}^{\infty} |k|\big|R(k)\big|.
\end{equation*}
Hence, we can find an appropriate constant $C > 0$ such that $R_1(k)$
is bounded by
\begin{equation*}
Cn^{-2}\sum_{j = 1}^{n} \int_{-1/2}^{1/2}\big|v - x_j\big|\,dv,
\end{equation*}
which both does not depend on $k$ and is easily seen to be
$O(n^{-1})$. This implies $\max_{k \in \Zint} R_1(k) = O(n^{-1})$.

Turning our attention to $R_2(k)$, we can assume without loss of
generality that $|k| > 0$ as this term is equal to zero whenever $k =
0$. The integral in $R_2(k)$ is equal to the sum of
\begin{equation*}
\int_{-1/2}^{1/2} \bigg\{ r\bigg(x_j + \frac{v - x_j}{2n + 1}\bigg)
 - r(x_j)\bigg\}\,dv\exp\big(-i2\pi kx_j\big)
\end{equation*}
and
\begin{equation*}
\int_{-1/2}^{1/2} \bigg\{
 r(x_j)\exp\big(-i2\pi kx_j\big)
 - r\bigg(x_j + \frac{v - x_j}{2n + 1}\bigg)
 \exp\bigg(-i2\pi k\bigg(x_j + \frac{v - x_j}{2n + 1}\bigg)\bigg)
 \bigg\}\,dv.
\end{equation*}
Therefore, we can see that $R_2(k)$ is bounded by the sum of $\max_{k
  \in \Zint} R_1(k)$, which we have already shown $\max_{k \in \Zint}
R_1 = O(n^{-1})$, and the quantity
\begin{align} \label{R2Bound}
\frac{1}{2n + 1} \sum_{j = 1}^{n} \bigg|
 \int_{-1/2}^{1/2} \bigg\{
 &r(x_j)\exp\big(-i2\pi kx_j\big) \\ \nonumber
 & - r\bigg(x_j + \frac{v - x_j}{2n + 1}\bigg)
 \exp\bigg(-i2\pi k\bigg(x_j + \frac{v - x_j}{2n + 1}\bigg)\bigg)
 \bigg\}\,dv\bigg|.
\end{align}
We can use the Fourier inversion formula to write
\begin{align} \label{FourierMagic}
&r(x_j)\exp\big(-i2\pi k x_j\big)
 - r\bigg(x_j + \frac{v - x_j}{2n + 1}\bigg)
 \exp\bigg(-i2\pi k\bigg(x_j + \frac{v - x_j}{2n + 1}\bigg)\bigg)
 \\ \nonumber
&= \sum_{\xi = -\infty}^{\infty} R(\xi) \bigg\{
 \exp\big(i2\pi(\xi - k)x_j\big)
 - \exp\bigg(i2\pi(\xi - k)\bigg(x_j + \frac{v - x_j}{2n + 1}\bigg)
 \bigg) \bigg\}.
\end{align}
From \eqref{FourierMagic} we can see that \eqref{R2Bound} is further
bounded by
\begin{align*}
&\frac{1}{2n + 1} \sum_{j = 1}^{n} \int_{-1/2}^{1/2}\,
 \sum_{\xi = -\infty}^{\infty} |R(\xi)|
 \bigg| \exp\big(i2\pi(\xi - k)x_j\big)
 - \exp\bigg(i2\pi(\xi - k)\bigg(x_j + \frac{v - x_j}{2n + 1}\bigg)
 \bigg) \bigg|\,dv \\
&\leq \frac{2^{3/2}\pi}{2n + 1}
 \bigg\{ \frac{1}{2n + 1} \sum_{j = 1}^{n}
 \int_{-1/2}^{1/2}\,|v - x_j|\,dv\bigg\}
 \bigg\{ \sum_{|\xi - k| > 0} |\xi - k||R(\xi)| \bigg\}.
\end{align*}
Since we have already shown $r \in \RR_{1}$, we have, for $\zeta
= \xi - k$,
$
\max_{|k| > 0} \sum_{|\zeta| > 0} |\zeta||R(k + \zeta)|
 < \infty.
$
Hence, we can find an appropriate constant $C > 0$ for \eqref{R2Bound}
to be further bounded by
\begin{equation*}
Cn^{-2} \sum_{j = 1}^{n} \int_{-1/2}^{1/2} |v - x_j|\,dv,
\end{equation*}
which both does not depend on $k$ and is easily seen to be of order
$O(n^{-1})$. Combining this fact with the result that $\max_{k \in
  \Zint} R_1(k) = O(n^{-1})$ implies $\max_{k \in \Zint} R_2(k) =
O(n^{-1})$.

Using the previous arguments we can also show $\max_{k \in \Zint}
R_3(k) = O(n^{-1})$ and $\max_{k \in \Zint} R_4(k) =
O(n^{-1})$. Finally, a similar argument for showing $\max_{k \in
\Zint} R_2(k) = O(n^{-1})$ can be used to show $\max_{k \in \Zint}
R_5(k) = O(n^{-1})$. This concludes the proof of Lemma
\ref{lemRhatOrder}.
\end{proof}


With the result of Lemma \ref{lemRhatOrder}, we can give the proof of
Lemma \ref{lemhthetabias} from Section \ref{mainresults1}:
\begin{proof}[Proof of Lemma \ref{lemhthetabias}]
We begin with the decomposition
\begin{equation*}
E\big[\htheta(x)\big] =
 \sum_{k = -\infty}^{\infty} \Lambda(h_nk) \Theta(k) \exp(i2\pi k x)
 + \sum_{k = -\infty}^{\infty} \frac{\Lambda(h_nk)}{\Psi(k)}
 \Big\{ E\big[\hR(k)\big] - R(k) \Big\} \exp(i2\pi k x)
\end{equation*}
so that $E[\htheta(x)] - \theta(x)$ is equal to
\begin{equation*}
\sum_{k = -\infty}^{\infty} \big\{\Lambda(h_nk) - 1\big\}
 \Theta(k) \exp(i2\pi k x)
 + \sum_{k = -\infty}^{\infty} \frac{\Lambda(h_nk)}{\Psi(k)}
 \Big\{ E\big[\hR(k)\big] - R(k) \Big\} \exp(i2\pi k x).
\end{equation*}
We can see that $\sup_{x \in [-1/2,\,1/2]} |E[\htheta(x)] -
\theta(x)|$ is bounded by
\begin{equation} \label{biasbound}
\sum_{k = -\infty}^{\infty} |\Lambda(h_nk) - 1||\Theta(k)|
 + \max_{k \in \Zint} \Big| E\big[\hR(k)\big] - R(k) \Big|
 \sum_{k = -\infty}^{\infty} \frac{|\Lambda(h_nk)|}{|\Psi(k)|}.
\end{equation}

Partition $\Zint$ into $I(h_n) \cup I^c(h_n)$, where $I(h_n) =
\{z \in \Zint\,:\,h_n|z| \leq M\} = \{z \in \Zint\,:\,|z| \leq
Mh_n^{-1}\}$. Hence, for every $k \in I^c(h_n)$, it follows that
$|\Lambda(h_nk)| \leq 1$, which implies both statements
$|\Lambda(h_nk) - 1| \leq 2$ and $|k| > Mh_n^{-1}$ hold. The first
term in the right--hand side of \eqref{biasbound} is therefore bounded
by
\begin{equation} \label{biassmoothbound}
2 \sum_{k \in I^c(h_n)} |\Theta(k)|
 \leq 2h_n^{s}M^{-s} \sum_{k = -\infty}^{\infty}
 |k|^s|\Theta(k)|.
\end{equation}
This implies the first term in \eqref{biasbound} is of the order
$O(h_n^{s})$, uniformly in $x \in [-1/2,\,1/2]$.

We now turn to the second term in \eqref{biasbound}. It follows from
Assumptions \ref{assumpPsi} and \ref{assumplambda} for the series in
this term to be bounded by
\begin{equation*}
h_n^{-1}\bigg[\min_{k \in \{z \in \Zint\,:\,|z| \leq \Gamma\}} |\Psi(k)|\bigg]^{-1}
 \Bigg\{h_n\sum_{\omega \in h_n\Zint} |\Lambda(\omega)|\Bigg\}
 + h_n^{-b-1}C_{\Psi}^{-1}
 \Bigg\{h_n\sum_{\omega \in h_n\Zint} |\omega|^{b}|\Lambda(\omega)|\Bigg\},
\end{equation*}
which is easily seen to be of the order $O(h_n^{-b - 1})$. The
additional factor of $h_n^{-1}$ appears in the bound above because we
have a shrinkage of $|k|$ by $h_n$. This implies $\sum_{k =
  -\infty}^{\infty} \{|\Lambda(h_nk)|/|\Psi(k)|\}$ is of the order
$O(h_n^{-b - 1})$. Now we only need to consider the term $\max_{k \in
  \Zint} |\hR(k) - R(k)|$. The assumptions of Lemma \ref{lemRhatOrder}
are satisfied. It then follows for $\max_{k \in \Zint} |\hR(k) - R(k)|
= O(n^{-1})$. Hence, the second term in \eqref{biasbound} is of the
order $O((nh_n^{b + 1})^{-1})$, uniformly in $x \in
[-1/2,\,1/2]$. Combining the results above, we have that
\eqref{biasbound} is of the order $O(h_n^{s} + (nh_n^{b + 1})^{-1})$,
uniformly in $x \in [-1/2,\,1/2]$, and the assertion of Lemma
\ref{lemhthetabias} follows.
\end{proof}


We are now prepared to state the proof of Lemma
\ref{lemthetahatconsistency}.

\begin{proof}[Proof of Lemma \ref{lemthetahatconsistency}]
Without loss of generality we can assume that $n \geq 3$. Our argument
is similar to the arguments found in Masry (1993), who gives related
results for an errors-in-variables model. We will employ truncation as
follows. Let the stabilizing sequence $\{\eta_n\}_{n \geq 3}$ satisfy
$\eta_n = O((nh_n^{2b + 1})^{-1/2}\log^{1/2}(n))$ and the truncation
sequence $\{t_n\}_{n \geq 3}$ satisfy $t_n =
O((n\log(n)(\log\log(n))^{1 + \delta})^{1/\kappa})$, with $\delta >
0$. Write $K_j = E^{1/\kappa}[|Y_j|^{\kappa}]$. We can decompose
$\htheta(x) - E[\htheta(x)]$ into the sum of $D_1(x) = \htheta(x) -
\htheta^t(x)$, $D_2(x) = E[\htheta^t(x)] - E[\htheta(x)]$ and $D_3(x)
= \htheta^t(x) - E[\htheta^t(x)]$, where
\begin{equation*}
\htheta^t(x) = \frac{1}{2n + 1} \sum_{j = -n}^{n}
 Y_j\1\big[|Y_j| \leq K_jt_n\big]W_{j,h_n}(x),
 \qquad x \in [-1/2,\,1/2].
\end{equation*}

Beginning with $D_1(x)$, it follows along the same lines as the
arguments in the proof of Lemma 2.1 of Masry (1993) for $\sup_{x \in
  [-1/2,\,1/2]} |D_1(x)| = o(\eta_n)$, almost surely. Turning our
attention now to $D_2(x)$, it is easy to show $\sup_{x \in
  [-1/2,\,1/2]} |W_{j,h_n}(x)|$ is bounded by the series $\sum_{k =
  -\infty}^{\infty} \{|\Lambda(h_nk)|/|\Psi(k)|\}$, and we have already
shown this series is of the order $O(h_n^{-b - 1})$ in the proof of
Lemma \ref{lemhthetabias}. Hence, we have that $\sup_{x \in
  [-1/2,\,1/2]} |W_{j,h_n}(x)| = O(h_n^{-b - 1})$. It follows that
we can find an appropriate constant $C > 0$ such that we can bound
$\sup_{x \in [-1/2,\,1/2]} |D_2(x)|$ by
\begin{equation} \label{D2Bound}
Ch_n^{-b - 1}\frac{1}{2n + 1} \sum_{j = -n}^{n}
 E\Big[|Y_j|\1\big[|Y_j| > K_jt_n\big]\Big].
\end{equation}
Since $\kappa > 1$, writing $M_K = \max_{j = -n,\ldots,n} K_j$, we
can apply Markov's inequality to obtain
\begin{equation*}
\max_{j = -n,\ldots,n} E[|Y_j|\1[|Y_j| > K_jt_n]] =
 \max_{j = -n,\ldots,n} \int_{0}^{\infty}
 P\big(|Y_j| > \max\{s,\,K_jt_n\}\big)\,ds
 \leq \frac{\kappa}{\kappa - 1}M_Kt_n^{1 - \kappa}.
\end{equation*}
Enlarging the constant $C$ in \eqref{D2Bound} implies
$\sup_{x \in [-1/2,\,1/2]} |D_2(x)| \leq Ch_n^{-b - 1}t_n^{1 - \kappa}
= o(\eta_n)$.

To continue, we will require an additional result. For any $u,v \in
[-1/2,\,1/2]$, we can repeat the arguments in the proof of Lemma
\ref{lemRhatOrder} to see that
\begin{equation*}
\Big|W_{j,h_n}(u) - W_{j,h_n}(v)\Big| \leq  \big|u - v\big|
 2^{3/2}\pi \sum_{k = -\infty}^{\infty} |k|
 \frac{|\Lambda(h_nk)|}{|\Psi(k)|}.
\end{equation*}
Hence, arguing as in the proof of Lemma \ref{lemhthetabias} we can
find an appropriate constant $C > 0$ such that
\begin{equation} \label{WjLipschitz}
\big|W_{j,h_n}(u) - W_{j,h_n}(v)\big| \leq Ch_n^{-b - 2}|u - v|,
 \qquad u,v \in [-1/2,\,1/2].
\end{equation}

Now we consider $D_3(x)$. Let $\{s_n\}_{n \geq 3}$ be a sequence
satisfying $s_n = O(h_n^{b + 2}\eta_nt_n^{-1}) = o(1)$ such that, when
we shatter the interval $[-1/2,\,1/2]$ into $s_n^{-1}$ many fragments
of the form $(x_i,\,x_{i + 1}]$, with the first fragment defined to be
$[-1/2,\,x_2] = \{-1/2\}\cup(-1/2,\,x_2]$, our fragments satisfy
$\max_{i = 1, \ldots, s_n^{-1}} |x_{i + 1} - x_i| \leq s_n$. For any
$x \in [-1/2,\,1/2]$, there is exactly one fragment $(x_{i'},\,x_{i' +
  1}]$ that contains $x$, and on this interval we can write
\begin{equation*}
D_3(x) = D_{4,i}(x) - D_{5,i}(x) + D_{6,i},
\end{equation*}
where $D_{4,i}(x) = \htheta^t(x) - \htheta^t(x_{i'})$, $D_{5,i}(x) =
E[\htheta^t(x)] - E[\htheta^t(x_{i'})]$ and $D_{6,i} =
\htheta^t(x_{i'}) - E[\htheta^t(x_{i'})]$. It follows that $\sup_{x
  \in [-1/2,\,1/2]} |D_3(x)|$ is bounded by
\begin{equation*}
\max_{i = 1,\ldots,s_n^{-1}} \sup_{x \in (x_i,\,x_{i + 1}]}
 \big|D_{4,i}(x)\big|
 + \max_{i = 1,\ldots,s_n^{-1}} \sup_{x \in (x_i,\,x_{i + 1}]}
 \big|D_{5,i}(x)\big|
 + \max_{i = 1,\ldots,s_n^{-1}} \big|D_{6,i}\big|.
\end{equation*}
Hence, to show the result $\sup_{x \in [-1/2,\,1/2]}
|D_3(x)| = O(\eta_n)$, almost surely, we will instead show that each
of the following statements hold:
\begin{equation} \label{D3TermBound1}
\max_{i = 1,\ldots,s_n^{-1}} \sup_{x \in (x_i,\,x_{i + 1}]} |D_{4,i}(x)|
 = O(\eta_n), \quad\text{a.s.},
\end{equation}
\begin{equation} \label{D3TermBound2}
\max_{i = 1,\ldots,s_n^{-1}}\sup_{x \in (x_i,\,x_{i + 1}]} |D_{5,i}(x)|
 = O(\eta_n)
\end{equation}
and
\begin{equation} \label{D3TermBound3}
\max_{i = 1,\ldots,s_n^{-1}} |D_{6,i}|
 = O(\eta_n), \quad\text{a.s.}.
\end{equation}

Beginning with \eqref{D3TermBound1}, fix an arbitrary interval
$(x_i,\,x_{i + 1}]$. On this interval $D_{4,i}(x)$ is equal to
\begin{equation*}
\frac{1}{2n + 1} \sum_{j = -n}^{n} Y_j \1\big[|Y_j| \leq K_jt_n\big]
 \big\{W_{j,h_n}(x) - W_{j,h_n}(x_i)\big\},
 \qquad x \in (x_i,\,x_{i + 1}].
\end{equation*}
It follows from \eqref{WjLipschitz} that we can find an appropriate
constant $C > 0$ for the inequality $\sup_{x \in (x_i,\,x_{i + 1}]}
|D_{4,i}(x)| \leq C t_nh_n^{-b - 2}s_n$ to hold, almost surely, independent of
$i$. Therefore, by construction of $\{s_n\}_{n \geq 3}$, we find that
\eqref{D3TermBound1} holds. Observing that $D_{5,i}(x) = E[D_{4,i}(x)]$, we have
that \eqref{D3TermBound2} holds as well.

To see the final statement \eqref{D3TermBound3} holds, define the random
variables $U_{j,i} = \{Y_j\1[|Y_j| \leq K_jt_n] - E[Y_j\1[|Y_j| \leq
K_jt_n]]\}W_{j,h_n}(x_i)$, $j = -n,\ldots,n$. Standard arguments can then
be used to show that $U_{-n,i},\ldots,U_{n,i}$ are independent, and
each have mean equal to zero, variance bounded by $C_1h_n^{-2b - 1}$
and bounded in absolute value by $C_2t_nh_n^{-b - 1}$, where $C_1 > 0$
and $C_2 > 0$ are appropriately chosen constants and both bounds are
independent of both $j$ and $i$. Applying Bernstein's Inequality (see,
for example, Lemma 2.2.11 in van der Vaart and Wellner, 1996), we can
find an appropriate constant $C > 0$ and obtain
\begin{equation} \label{D6ProbBound}
P\bigg( \max_{i = 1,\ldots,s_n^{-1}} |D_{6,i}| > \eta_n \bigg)
 \leq 2s_n^{-1}\exp\bigg(-C\frac{n\eta_n^2}{h_n^{-2b - 1} + t_nh_n^{-b - 1}\eta_n}
 \bigg).
\end{equation}
In light of the fact that $t_nh_n^{-b - 1}\eta_n = o(h_n^{-2b - 1})$,
which holds since $\kappa > 2 + 1/b$, we can enlarge $C$ for the
right--hand side of \eqref{D6ProbBound} to be further bounded by a
positive constant multiplied by
\begin{equation*}
h_n^{-3/2}n^{(1/2) + (1/\kappa) - C} \log^{-(1/2 - 1/\kappa)}(n)
 \big(\log\log(n)\big)^{(1 + \delta)/\kappa},
\end{equation*}
which is summable provided we take $C > (3/2)(1 + 1/(2b + 1)) +
1/\kappa$, where $1/(2b + 1)$ accounts for the expansion of
$h_n^{-1}$; i.e.\ $(n^{1/(2b + 1)}h_n)^{-3/2} \to 0$, as $n \to
\infty$. It then follows by the Borel--Cantelli lemma that
\eqref{D3TermBound3} holds. This completes the proof.
\end{proof}


We can now state the proof of Theorem \ref{thmthetahatuniformrate}
from Section \ref{mainresults1}:
\begin{proof}[Proof of Theorem \ref{thmthetahatuniformrate}]
The first two assertions follow immediately from the results of Lemma
\ref{lemhthetabias} and Lemma \ref{lemthetahatconsistency} in
combination with our choice of regularizing sequence as discussed in
Section \ref{mainresults1}. This means we only need to show the last
assertion. Let us begin by calculating the Fourier coefficients
$\{\hat\Theta(\xi)\}_{\xi \in \Zint}$ of $\htheta$:
\begin{align*}
\hat \Theta(\xi) &= \int_{-1/2}^{1/2} \htheta(x) e^{-i2\pi\xi x}\,dx
 = \sum_{k = -\infty}^{\infty} \frac{\Lambda(h_nk)}{\Psi(k)}\hR(k)
 \int_{-1/2}^{1/2}e^{i2\pi(k - \xi)x}\,dx \\
&= \Lambda(h_n\xi)\Theta(\xi)
 + \Bigg\{ E\big[\hR(\xi)\big] - R(\xi)
 + \frac{1}{2n + 1} \sum_{j = -n}^{n} \ve_je^{-i2\pi\xi x_j}\Bigg\}
 \frac{\Lambda(h_n\xi)}{\Psi(\xi)},
\end{align*}
where we have used the orthonormality of the basis
$\{\exp(i2\pi k x)\,:\,x \in [-1/2,\,1/2]\}_{k \in \Zint}$ in the
final equality. The definition of $\RR_{s - 1/2}$ requires that we
show the series condition
\begin{equation} \label{SeriesCondition}
\sum_{\xi = -\infty}^{\infty} |\xi|^{s - 1/2} \big|\hat \Theta(\xi)\big|
 < \infty
\end{equation}
is satisfied. We can see that $|\hat \Theta(\xi)|$ is bounded by
\begin{equation} \label{ThetahatabsBound}
|\Theta(\xi)|
 + \Bigg\{ \max_{k \in \Zint}\Big|E\big[\hR(k)\big] - R(k)\Big| 
 + \bigg|\frac{1}{2n + 1} \sum_{j = -n}^{n} \ve_je^{i2\pi\xi x_j}\bigg|
 \Bigg\} \frac{|\Lambda(h_n\xi)|}{|\Psi(\xi)|}.
\end{equation}
Observing that $\theta \in \RR_{s}$ implies $\theta \in \RR_{s -
  1/2}$, we have $\sum_{\xi = -\infty}^{\infty} |\xi|^{s - 1/2}
|\Theta(\xi)| < \infty$. Hence, we only need to verify the series
condition \eqref{SeriesCondition} stated for the last term in
\eqref{ThetahatabsBound} holds.

The assumptions of Lemma \ref{lemRhatOrder} are satisfied, which
gives $\max_{k \in \Zint}|E[\hR(k)] - R(k)| = O(n^{-1})$. Additionally,
the map $x \mapsto \exp(-i2\pi k x)$ is confined to the unit circle in
the complex plane. A standard argument then shows
\begin{equation*}
\max_{k \in \Zint} \bigg|
 \frac{1}{2n + 1} \sum_{j = -n}^{n} \ve_je^{-i2\pi k x_j}
 \bigg| = O\big(n^{-1/2}\log^{1/2}(n)\big),
 \qquad\text{a.s.}
\end{equation*}
Finally, in the proof of Lemma \ref{lemhthetabias}, we have shown
$\sum_{k = -\infty}^{\infty} \{|\Lambda(h_nk)|/|\Psi(k)|\} = O(h_n^{-b
  - 1})$, and similar lines of argument can be used to show $\sum_{k =
  -\infty}^{\infty} \{|k|^{s - 1/2}|\Lambda(h_nk)|/|\Psi(k)|\} = O(h_n^{-s - b
  - 1/2})$ with the assumption $\int_{-\infty}^{\infty}\,|u|^{s +
  b - 1/2}|\Lambda(u)|\,du < \infty$. Together, these results imply
\begin{equation*}
\Bigg\{ \max_{k \in \Zint}\Big| E\big[\hR(k)\big] - R(k) \Big|
 + \max_{k \in \Zint} \bigg|
 \frac{1}{2n + 1} \sum_{j = -n}^{n} \ve_je^{-i2\pi k x_j} \bigg| \Bigg\}
 \sum_{\xi = -\infty}^{\infty} |\xi|^{s - 1/2}
 \frac{|\Lambda(h_n\xi)|}{|\Psi(\xi)|}
\end{equation*}
is of order $O(1 + (n\log(n))^{-1/2}) = O(1)$, almost surely. Hence,
the series condition \eqref{SeriesCondition} stated for the last term
in \eqref{ThetahatabsBound} holds. It follows that $\htheta - \theta
\in \RR_{s - 1/2}$, almost surely, for large enough $n$. Combining
this statement with the first assertion then proves the third
assertion.
\end{proof}


Nickl and P\"otscher (2007) study classes of functions of Besov- and
Sobolev-type. These authors derive results concerning the bracketing
metric entropy and the related central limit theorems of these
spaces using weighted norms. Since our space $\RR_{s}$ is a collection
of functions with compact support on the interval $[-1/2,\,1/2]$, we
can see the results of their Corollary 4 on bracketing numbers for
weighted Sobolev spaces immediately apply to our case by repeating the
steps in the proof of their Corollary 2 for Besov-type functions of
bounded support, i.e.\ our function space $\RR_{s,1}$ is the unit ball
of the metric space $(\RR_{s},\,\|\cdot\|_{\infty})$ and $\RR_{s,1}$
can be viewed as a restriction of a larger weighted Sobolev space of
similar type, where the weighting function is now defined to be equal
to 1 on the interval $[-1/2,\,1/2]$. We can summarize this result in
the following proposition:
\begin{proposition}
 \label{corRsEntropy}
For the function space $\RR_{s,1}$, with $s > 1/2$, a finite constant
$C > 0$ exists such that
\begin{equation*}
\log N_{[\,]}\big(\epsilon,\,\RR_{s,1},\,\|\cdot\|_{\infty}\big)
 \leq C\epsilon^{-1/s},
 \qquad \epsilon > 0,
\end{equation*}
where $N_{[\,]}(\epsilon,\,\RR_{s,1},\,\|\cdot\|_{\infty})$ is the
number of brackets of length $\epsilon$ required to cover the metric
space $(\RR_{s,1},\,\|\cdot\|_{\infty})$.
\end{proposition}


In light of the results on the estimator $\htheta$, we can now state a
result on the modulus of continuity relating $\hFemp(t)$ to $(2n +
1)^{-1} \sum_{j = -n}^{n} \1[\ve_j \leq t]$. Using results on Donsker
classes of functions, we can show this modulus of continuity holds up
to a negligible term of order $o_P(n^{-1/2})$. The proof of this
result follows along the same lines as the proof of Lemma A.1 in Van
Keilegom and Akritas (1999) and, therefore, it is omitted.
\begin{lemma} \label{lemFhatModulus}
Let the assumptions of Theorem \ref{thmthetahatuniformrate} be
satisfied with $s > 3/2$. In addition, assume that $F$ admits a bounded
Lebesgue density function $f$. Then $\sup_{t \in \R} |M_n(t)| = \opn$,
where
\begin{align*}
M_n(t) &= \frac{1}{2n + 1} \sum_{j = -n}^{n}
 \1\big[\ve_j \leq t + \big[K\big(\htheta - \theta\big)\big](x_j)\big]
 - \int_{-1/2}^{1/2}
 F\big(t + \big[K\big(\htheta - \theta\big)\big](x)\big)\,dx \\
&\quad - \frac{1}{2n + 1} \sum_{j = -n}^{n} \1[\ve_j \leq t]
 + F(t).
\end{align*}
\end{lemma}


Direct regression estimators typically allow for appropriate
expansions into averages of the model errors up to some negligible
remainder term. This representation motivates the term $\ve f(t)$ in
the expansion of the empirical distribution function of the these
model residuals. In the following result, we provide a similar
expansion for the indirect regression estimator $\htheta$, and we show
this expansion holds up to a negligible term of order $\opn$. Hence,
we can immediately see that our indirect regression function estimator
$\htheta$ and typical direct regression function estimators share this
property. This combined with the modulus of continuity result above
implies that our residual-based empirical distribution function behaves
similarly to that in the usual direct estimation setting (see, for
example, M\"uller, Schick and Wefelmeyer, 2007, who construct
expansions for many residual-based empirical distribution functions
based on direct regression function estimators).

\begin{proposition} \label{propKthetahatExpan}
Let the assumptions of Lemma \ref{lemhthetabias} be satisfied, and
assume that $E[\ve_j^2] < \infty$, $j = -n,\ldots,n$. Let the
regularizing sequence $\{h_n\}_{n \geq 1}$ satisfy $h_n^{s + b + 1} =
o(n^{-1/2})$ and $(nh_n)^{-1} = o(n^{-1/2})$. Then
\begin{equation*}
\bigg| \int_{-1/2}^{1/2} \big[K\big(\htheta - \theta\big)\big](x)\,dx
 - \frac{1}{2n + 1} \sum_{j = -n}^{n} \ve_j \bigg|
 = \opn.
\end{equation*}
\end{proposition}
\begin{proof}
Note that $\hR(k) - E[\hR(k)] = (2n + 1)^{-1} \sum_{j = -n}^{n}
\ve_j\exp(-i2\pi k x_j)$. We can write $1 = \int_{-1/2}^{1/2} \sum_{k
  = -\infty}^{\infty} e^{i2\pi k x}\,dx$ so that we can bound the
left--hand side of the assertion by $S_1 + S_2 + S_3$, where
\begin{equation*}
S_1 = \bigg|\frac{1}{2n + 1} \sum_{j = -n}^{n} \ve_j
 \int_{-1/2}^{1/2} \bigg\{\sum_{k = -\infty}^{\infty}
 \big\{\Lambda(h_nk) - 1\big\} e^{i2\pi k(x - x_j)}\bigg\}\,dx \bigg|,
\end{equation*}
\begin{equation*}
S_2 = \sum_{k = -\infty}^{\infty}
 \big|\Lambda(h_nk) - 1\big|\big|R(k)\big|
 \bigg|\int_{-1/2}^{1/2} e^{i2\pi k x}\,dx\bigg|.
\end{equation*}
and
\begin{equation*}
S_3 = \bigg[\max_{k \in \Zint}\Big|E\big[\hR(k)\big] - R(k)\Big|\bigg]
 \sum_{k = -\infty}^{\infty} |\Lambda(h_nk)|
\end{equation*}
The assertion then follows, if we show $S_1 = \opn$, $S_2 =
o(n^{-1/2})$ and $S_3 = o(n^{-1/2})$.

We can see that it follows for $S_1 = \opn$ from Assumption
\ref{assumplambda} and the fact that
\begin{equation*}
\frac{1}{2n + 1} \sum_{j = -n}^{n} \bigg\{
 \int_{-1/2}^{1/2} \bigg\{ \sum_{k = -\infty}^{\infty}
 \big\{\Lambda(h_nk) - 1\big\}e^{i2\pi k(x - x_j)} \bigg\}\,dx
 \bigg\}^{2} = o(1).
\end{equation*}

To show that $S_2 = o(n^{-1/2})$, recall the convolution theorem for
Fourier transformation implies that $|R(k)| =
|\Theta(k)||\Psi(k)|$. The integral term in $S_2$ is bounded by a
positive constant $C$ multiplied by $|k|^{-1}$. Combining this fact
with the constant $C_{\Psi}^*$ from Assumption \ref{assumpPsi} and
that $|\Lambda(h_nk) - 1| \leq 2$ shows that we can enlarge $C$ such
that $S_2$ is bounded by
\begin{equation*}
C h_n^{s + b + 1} \sum_{k = -\infty}^{\infty} |k|^s|\Theta(k)|
 = O\big(h_n^{s + b + 1}\big) = o\big(n^{-1/2}\big).
\end{equation*}

Finally, we consider the last remainder term $S_3$. The assumptions of
Lemma \ref{lemRhatOrder} are satisfied, which gives $\max_{k \in
  \Zint} |E[\hR(k)] - R(k)| = O(n^{-1})$. Similar lines of argument to
those in the proof of Lemma \ref{lemhthetabias} shows the series term
in $S_3$ is of the order $O(h_n^{-1})$. Combining these facts, we have
that $S_3$ is of the order $O((nh_n)^{-1}) = o(n^{-1/2})$, which
concludes the proof.
\end{proof}


Combining the results above, we can now state the proof of Theorem
\ref{thmFhatExpan}.

\begin{proof}[Proof of Theorem \ref{thmFhatExpan}]
Recall $M_n(t)$ from Lemma \ref{lemFhatModulus}. A straightforward
calculation shows that
\begin{equation*}
\frac{1}{2n + 1} \sum_{j = -n}^{n}
 \Big\{ \1\big[\hve_j \leq t\big] - \1\big[\ve_j \leq t\big]
 - \ve_jf(t) \Big\}  = 
 M_n(t) + H_n(t) + L_n(t),
\end{equation*}
where
\begin{equation*}
H_n(t) = \int_{-1/2}^{1/2}
 F\big(t + \big[K\big(\htheta - \theta\big)\big](x)\big)
 \,dx - F(t)
 -f(t)\int_{-1/2}^{1/2} \big[K\big(\htheta - \theta\big)\big](x)\,dx
\end{equation*}
and
\begin{equation*}
L_n(t) = f(t)\bigg\{
 \int_{-1/2}^{1/2} \big[K\big(\htheta - \theta\big)\big](x)\,dx
 - \frac{1}{2n + 1} \sum_{j = -n}^{n} \ve_j \bigg\}.
\end{equation*}
The assumptions of Lemma \ref{lemFhatModulus} are satisfied, which
implies $\sup_{t \in \R}|M_n(t)| = \opn$. Hence, the assertion follows
from showing $\sup_{t \in \R} |H_n(t)| = \opn$ and $\sup_{t \in \R}
|L_n(t)| = \opn$.

Beginning with $H_n(t)$, writing $C_{f,\gamma}$ for the H\"older
constant of $f$ with exponent $\gamma$, we have
\begin{equation*}
H_n(t) = \int_{-1/2}^{1/2} \big[K\big(\htheta - \theta\big)\big](x)
 \int_{0}^{1}
 \Big\{f\big(t + s\big[K\big(\htheta - \theta\big)\big](x)\big)
 - f(t)\Big\}\,ds\,dx
\end{equation*}
so that $\sup_{t \in \R} |H_n(t)|$ is bounded by
\begin{equation*}
\frac{C_{f,\gamma}}{1 + \gamma} \bigg[
 \sup_{x \in [-1/2,\,1/2]} \Big|\htheta(x) - \theta(x)\Big|
 \bigg]^{1 + \gamma}.
\end{equation*}
The assumptions of Theorem \ref{thmthetahatuniformrate} are satisfied,
which implies the second term in the bound above is $o(n^{-1/2})$,
almost surely. It then follows that $\sup_{t \in \R} |H_n(t)| = \opn$.

Now we will consider $L_n(t)$. Since $f$ is bounded, we have that
$\sup_{t \in \R} |L_n(t)|$ is bounded by
\begin{equation*}
\sup_{t \in \R}|f(t)|\bigg|
 \int_{-1/2}^{1/2} \big[K\big(\htheta - \theta\big)\big](x)\,dx
 - \frac{1}{2n + 1} \sum_{j = -n}^{n} \ve_j \bigg|.
\end{equation*}
The parameter sequence $\{h_n\}_{n \geq 1}$ satisfies \eqref{bw}, and
we have both
\begin{equation*}
h_n^{s + b + 1}
 = O\big(n^{-1/2 - 1/(4s + 4b + 2)}
 \log^{(s + b + 1)/(2s + 2b + 1)}(n)\big)
 = o\big(n^{-1/2}\big)
\end{equation*}
and
\begin{equation*}
\big(nh_n\big)^{-1}
 = O\big(n^{-(2s + 2b)/(2s + 2b + 1)}
 \log^{-1/(2s + 2b + 1)}(n)\big)
 = o\big(n^{-1/2}\big),
\end{equation*}
which follows from the fact that $s + b > 1/2$. The assumptions of
Proposition \ref{propKthetahatExpan} are satisfied, which implies the
second term in the bound above is $\opn$. This shows that $\sup_{t \in
  \R}|L_n(t)| = \opn$, and, hence, the assertion of Theorem
\ref{thmFhatExpan} holds.
\end{proof}

Here we provide a short proof of Proposition \ref{prophthetaVarBiassq}.

\begin{proof}[Proof of Proposition \ref{prophthetaVarBiassq}] 
Beginning with the first assertion, we can write the integrated
variance of $\htheta$ as
\begin{equation*}
\int_{-1/2}^{1/2} E\Big[\big\{\htheta(x)
 - E\big[\htheta(x)\big]\big\}^2\Big]\,dx
 = \frac{\sigma^2}{2n + 1} \sum_{k = -\infty}^{\infty}
 \frac{\Lambda^2(h_nk)}{\Psi^2(k)}.
\end{equation*}
Repeating the arguments in the proof of Lemma \ref{lemhthetabias} in
Section \ref{details} then shows
\begin{equation*}
\sum_{k = -\infty}^{\infty} \frac{\Lambda^2(h_nk)}{\Psi^2(k)} =
O\big(h_n^{-2b - 1}\big).
\end{equation*}
Therefore, we can specify $C_{\Lambda} > 0$ for the first assertion to
hold. The second assertion follows directly by an application of Lemma
\ref{lemhthetabias}.
\end{proof}

Now returning to the discussion in Section \ref{boot}, the choice of
scaling sequence $\{c_n\}_{n \geq 1}$ used for the contaminates
$c_nU_j$, $j=-n,\ldots,n$, in the smooth bootstrap, always satisfies
$(nc_n)^{-1}\log(n) = o(1)$, and, hence, we can apply Theorem A of
Silverman (1978) in combination with the H\"older continuity of $w$
and the results of Theorem \ref{thmthetahatuniformrate}, or the
combination of the results from Lemma \ref{lemhthetabias} and Lemma
\ref{lemthetahatconsistency}, to show that $f_n^*$ is strongly
consistent for $f$, uniformly over the entire real line. The result
only holds when the density function $f$ is H\"older with exponent
$2/3 < \gamma \leq 1$, the density function $w$ is also chosen to be
H\"older continuous with similar smoothness, and the smoothness index
$s$ of the function space $\RR_{s}$ satisfies $s > (1 + \gamma)(2b +
1)/(3\gamma - 2) > 4b + 2$, which is twice the lower bound on $s$
required by the second statement of Theorem
\ref{thmthetahatuniformrate}. The additional smoothness in $\theta$
is required due to the fact that residuals are used in the estimator
$f_n^*$ rather than the model errors. Under these conditions, further
technical but standard arguments similar to those used to prove
related results in Neumeyer (2009) can be used to prove the following
result.
\begin{proposition} \label{propGausCovTerms}
Let the assumptions of Theorem \ref{thmFhatbootExpan} be satisfied,
but now requiring the density $f$ to be H\"older continuous with
exponent $2/3 < \gamma \leq 1$ and choosing the density $w$ also to be
H\"older continuous with the same exponent. Additionally, assume the
smoothness index $s$ of the function space $\RR_{s}$ satisfies $s > (1
+ \gamma)(2b + 1)/(3\gamma - 2)$. Then
\begin{equation*}
\sup_{t \in \R} \Big| f_{n}^*(t) - f(t) \Big| = \op,
\end{equation*}
\begin{equation*}
\sup_{t \in \R} \Big| F_{n}^*(t) - F(t) \Big| = \op
\end{equation*}
and
\begin{equation*}
\sup_{t \in \R} \Big| E^*\big[\ve^*\1\big[\ve^* \leq t\big]\big]
 - E\big[\ve\1[\ve \leq t]\big] \Big| = \op.
\end{equation*}
\end{proposition}

We omit proof of the following result because it is proven in exactly
the same manner as Theorem \ref{thmthetahatuniformrate}.
\begin{proposition} \label{prophthetabootuniformrate}
Let the assumptions of Theorem \ref{thmthetahatuniformrate} be
satisfied. Choose the regularizing sequence $\{g_n\}_{n \geq 1}$
according to \eqref{bw} and let the scaling sequence $\{c_n\}_{n \geq
  1}$ satisfy $c_n = O(n^{-\alpha})$, with $0 < \alpha < 1/2 +
1/\kappa$. Then, $P^*$--outer almost surely, we have
\begin{equation*}
\sup_{x \in [-1/2,\,1/2]} \Big| \htheta^*(x) - \htheta(x) \Big|
 = O\big(n^{-s/(2s + 2b + 1)}\log^{s/(2s + 2b + 1)}(n)\big),
\end{equation*}
if, additionally, $s > (2b + 1)/(2\gamma)$, for some $0 < \gamma \leq
1$,
\begin{equation*}
\bigg[\sup_{x \in [-1/2,\,1/2]} \Big| \htheta^*(x) - \htheta(x) \Big|
 \bigg]^{1 + \gamma} = o(n^{-1/2}),
\end{equation*}
and, for large enough $n$,
\begin{equation*}
\htheta^* - \htheta \in \RR_{s - 1/2,1}.
\end{equation*}
\end{proposition}



\begin{thebibliography}{35}

\bibitem{A1995}
Adorf, H.M.\ (1995).
Hubble Space Telescope image reconstruction in its fourth year.
{\it Inverse Problems} {\bf 11}, 639-653.

\bibitem{BBDV2009}
Bertero, M., Boccacci, P., Desider\`a, G.\ and Vicidomini, G.\
(2009).
Image deblurring with Poisson data: from cells to galaxies.
{\it Inverse Problems} {\bf 25}, 123006.

\bibitem{BBH2010}
Birke, M., Bissantz, N.\ and Holzmann, H.\ (2010).
Confidence bands for inverse regression models.
{\it Inverse Problems} {\bf 26}, 115020.

\bibitem{BH2008}
Bissantz, N.\ and Holzmann, H.\ (2008).
Statistical inference for inverse problems.
{\it Inverse Problems} {\bf 24}, 034009.

\bibitem{BH2013}
Bissantz, N.\ and Holzmann, H.\ (2013).
Asymptotics for spectral regularization estimators in statistical
inverse problems.
{\it Computat. Statist.} {\bf 28}, 435-453.

\bibitem{BMM2006}
Bissantz, N., Mair, B.\ and Munk, A.\ (2006).
A multi-scale stopping criterion for MLEM reconstructions in PET.
{\it IEEE Nuclear Science Symposium Conference Record} {\bf 6}, 3376-3379.

\bibitem{BMM2008}
Bissantz, N., Mair, B.\ and Munk, A.\ (2008).
A statistical stopping rule for MLEM reconstructions in PET.
{\it IEEE Nuclear Science Symposium Conference Record} {\bf 8}, 4198-4200.

\bibitem{BHR2016}
Blanchard, G., Hoffmann, M.\ and Rei\ss, M.\ (2016).
Optimal adaptation for early stopping in statistical inverse problems.
arXiv:1606.07702v1.

\bibitem{BM2012}
Blanchard, G.\ and Math\'e, P.\ (2012).
Discrepancy principle for statistical inverse problems with
application to conjugate gradient iteration.
{\it Inverse Problems} {\bf 28}, 115011.

\bibitem{C1993}
Cao, R.\ (1993).
Bootstrapping the mean integrated squared error.
{\it J. Multivariate Anal.} {\bf 45}, 137-160.

\bibitem{C2000}
Cavalier, L.\ (2000).
Efficient estimation of a density in a problem of tomography.
{\it Ann. Statist.} {\bf 28}, 630-647.

\bibitem{C2008}
Cavalier, L.\ (2008).
Nonparametric statistical inverse problems.
{\it Inverse Problems} {\bf 24}, 034004.

\bibitem{CG2006}
Cavalier, L.\ and Golubev, Y.\ (2006).
Risk hull method and regularization by projections of ill-posed
inverse problems.
{\it Ann. Statist.} {\bf 34}, 1653-1677.

\bibitem{CT2002}
Cavalier, L.\ and Tsybakov, A.\ (2002).
Sharp adaptation for inverse problems with random noise.
{\it Probab. Theory Related Fields} {\bf 123}, 323-354.

\bibitem{DM2008}
Davies, P.L.\ and Meise, M.\ (2008).
Approximating data with weighted smoothing splines.
{\it J. Nonparametr. Stat.} {\bf 20}, 207-228.

\bibitem{F1991}
Fan, J.\ (1991).
On  the optimal rates of convergence for nonparametric deconvolution
problems.
{\it Ann. Statist.} {\bf 19}, 1257-1272.

\bibitem{G1999}
Goldenshluger, A.\ (1999).
On pointwise adaptive nonparametric deconvolution.
{\it Bernoulli} {\bf 5}, 907-925.

\bibitem{GMMMPG2004}
Gonz\'alez-Manteiga, W., Martinez-Miranda, M.D.\ and
P\'erez-Gonz\'alez, A.\ (2004).
The choice of smoothing parameter in nonparametric regression through
wild bootstrap.
{\it Comput. Statist. Data Anal.} {\bf 47}, 487-515.

\bibitem{HH2005}
Hall, P.\ and Horowitz, J.\ (2005).
Nonparametric methods for inference in the presence of instrumental
variables.
{\it Ann. Statist.} {\bf 33}, 2904-2929.

\bibitem{HMSDKM2012}
Hotz, T., Marnitz, P., Stichtenoth, R., Davies, L., Kabluchko, Z.\
and Munk, A.\ (2012).
Locally adaptive image denoising by a statistical multiresolution
criterion.
{\it Comput. Statist. Data Anal.} {\bf 56}, 543-558.

\bibitem{MR1996}
Mair, B.A.\ and Ruymgaart, F.H.\ (1996).
Statistical inverse estimation in Hilbert scales.
{\it SIAM J. Appl. Math.} {\bf 56}, 1424-1444.

\bibitem{MM2014}
Marteau, C.\ and Math\'e, P.\ (2014).
General regularization schemes for signal detection in inverse
problems.
{\it Math. Methods Statist.} {\bf 23}, 176-200.

\bibitem{M1991}
Masry, E.\ (1991).
Multivariate probability density deconvolution for stationary random
processes.
{\it IEEE Trans. Inform. Theory} {\bf 37}, 1105-1115.

\bibitem{M1993}
Masry, E.\ (1993).
Multivariate regression estimation with errors-in-variables for
stationary processes.
{\it J. Nonparametr. Stat.} {\bf 3}, 13-36.

\bibitem{MP2006}
Math\'e, P.\ and Pereverzev, S.V.\ (2006).
Regularization of some linear ill-posed problems with discretized
random noisy data.
{\it Math. Comp.} {\bf 75}, 1913-1929.

\bibitem{MSW2004}
M\"uller, U.U., Schick, A.\ and Wefelmeyer, W.\ (2004).
Estimating linear functionals of the error distribution in
nonparametric regression.
{\it J. Statist. Plann. Inference} {\bf 119}, 75-93.

\bibitem{MSW2007}
M\"uller, U.U., Schick, A.\ and Wefelmeyer, W.\ (2007).
Estimating the error distribution function in semiparametric
regression.
{\it Statist. Decisions} {\bf 25}, 1-18.

\bibitem{N2009}
Neumeyer, N.\ (2009).
Smooth residual bootstrap for empirical processes of non-parametric
regression residuals.
{\it Scand. J. Stat.} {\bf 36}, 204-228.

\bibitem{NP2007}
Nickl, R.\ and P\"otscher, B.M.\ (2007).
Bracketing metric entropy rates and empirical central limit theorems
for function classes of Besov- and Sobolev-type.
{\it J. Theoret. Probab.} {\bf 20}, 177-199.

\bibitem{P1962}
Parzen, E.\ (1962).
On estimation of a probability density function and mode.
{\it Ann. Math. Statist.} {\bf 33}, 1065-1076.

\bibitem{PR1999}
Politis, D.N.\ and Romano, J.P.\ (1999).
Multivariate density estimation with general flat-top kernels of
infinite order.
{\it J. Multivariate Anal.} {\bf 68}, 1-25.

\bibitem{PBD2015}
Proksch, K., Bissantz, N.\ and Dette, H.\ (2015).
Confidence bands for multivariate and time dependent inverse
regression models.
{\it Bernoulli} {\bf 21}, 144-175.

\bibitem{S1978}
Silverman, B.W.\ (1978)
Weak and strong uniform consistency of the kernel estimate of a
density and its derivatives.
{\it Ann. Statist.} {\bf 6}, 177-184.

\bibitem{S1986}
Silverman, B.W.\ (1986).
{\it Density estimation for statistics and data analysis.}
Vol. 26. CRC press.

\bibitem{VW1998}
van der Vaart, A.W.\ and Wellner J.A.\ (1996).
{\it Weak convergence and empirical processes. With applications to 
statistics.} Springer Series in Statistics. Springer-Verlag, New York.

\bibitem{VKA1999}
Van Keilegom, I.\ and Akritas, M.G.\ (1999).
Transfer of tail information in censored regression models. 
{\it Ann. Statist.} {\bf 27}, 1745-1784.

\end{thebibliography}
\end{document}